\documentclass[journal]{IEEEtran}

\usepackage{amssymb}
\usepackage{amsmath}
\usepackage{cite}
\usepackage{url}
\usepackage{empheq}
\usepackage{xcolor}
\usepackage{graphicx}
\usepackage{subfigure}
\usepackage{enumitem}
\usepackage{fancyhdr}
\usepackage{mdwmath}	
\usepackage{mdwtab}
\usepackage{caption}
\usepackage{amsthm}
\usepackage{algorithm}
\usepackage{algorithmic}

\newtheorem{lemma}{Lemma}
\newtheorem{remark}{Remark}
\newtheorem{theorem}{Theorem}
\newtheorem{corollary}{Corollary}
\newtheorem{assumption}{Assumption}

\newtheorem{proposition}{Proposition}

\newcommand{\eqr}[1]{(\ref{#1})}
\newcommand{\fref}[1]{Fig.~\ref{#1}}

\newcommand*{\QEDA}{\null\nobreak\hfill\ensuremath{\blacksquare}}

\begin{document}
\title{Leaky-Coaxial Pinching-Antenna System with Adjustable Slot Apertures}
\author{Kaidi~Wang,~\IEEEmembership{Member,~IEEE,}
Daniel~K.~C.~So,~\IEEEmembership{Senior~Member,~IEEE,}
Zhiguo~Ding,~\IEEEmembership{Fellow,~IEEE,}
and~George~K.~Karagiannidis,~\IEEEmembership{Fellow, IEEE}
\thanks{Kaidi~Wang and Daniel~K.~C.~So are with the Department of Electrical and Electronic Engineering, the University of Manchester, Manchester, M1 9BB, UK (email: kaidi.wang@ieee.org; d.so@manchester.ac.uk).}
\thanks{Zhiguo~Ding is with the School of Electrical and Electronic Engineering (EEE), Nanyang Technological University, Singapore 639798 (e-mail: zhiguo.ding@ntu.edu.sg).}
\thanks{George~K.~Karagiannidis is with Department of Electrical and Computer Engineering, Aristotle University of Thessaloniki, Greece (email: geokarag@auth.gr).}}
\maketitle
\begin{abstract}
As a practical physical implementation of pinching-antenna systems, leaky coaxial cable (LCX) enables distributed radiation in more general wireless environments, particularly for lower-frequency applications. In this paper, a leaky-coaxial pinching-antenna system, referred to as the LCX pinching-antenna system, is investigated, and adjustable slot apertures are introduced, such that the slot size can be continuously adjusted rather than being restricted to binary activation. Specifically, the aperture adjustment is modeled as amplitude scaling of the channels induced by the corresponding slots, or equivalently, as power coefficients associated with different slots. Accordingly, analytical results are derived to quantify the performance gain of continuous aperture adjustment over binary slot activation and to reveal the impact of channel coherence on the achievable data rate improvement. Furthermore, static and dynamic time-division multiple access (TDMA) schemes are considered, and the corresponding sum rate maximization problems are formulated and efficiently solved by quadratic transform based optimization, combined with successive convex approximation and alternating updates. Simulation results demonstrate that the proposed design can significantly outperform conventional fixed-antenna systems, traditional LCX schemes, and binary slot activation in terms of both achievable sum rate and outage probability.
\end{abstract}
\section{Introduction}
With the rapid advancement of wireless communications, the number of connected devices has increased dramatically, leading to increasingly complex propagation environments. In such scenarios, conventional fixed-antenna deployments often encounter difficulties in providing reliable coverage and maintaining favorable channel conditions. To address these challenges, various flexible antenna technologies have been proposed to enhance the adaptability of wireless systems, such as reconfigurable intelligent surfaces (RIS), fluid antennas, and movable antennas \cite{wu2019irs, wong2020fluid, zhu2023modeling}. These technologies aim to improve wireless communication performance through environmental reconfiguration or antenna flexibility, thereby enabling more controllable signal propagation and enhanced channel conditions \cite{wu2021intelligent, wu2024fluid, ma2024movable}. As one of the most recent members of this family of flexible antenna technologies, pinching-antenna systems have attracted considerable attention due to their capability of directly adjusting channel amplitudes, adaptively establishing line-of-sight (LoS) communication links, and achieving low-cost implementation \cite{suzuki2022pinching, ding2024pin}. These features make pinching antennas a promising paradigm for future wireless networks \cite{liu2025pinching, yang2025pinching}.

A typical pinching-antenna system consists of a guiding medium, such as a waveguide, along which the transmitted signal propagates, and a set of radiating elements, such as pinches, that radiate the guided signal into free space \cite{ding2025edma, tegos2025pin}. By selectively deploying pinches at designated locations along the waveguide, the guided signal can be flexibly leaked into the wireless environment, thereby enabling controllable distributed transmission \cite{kaidi2025pin4, xie2026pinching}. Depending on the implementation of the pinching points, pinching-antenna systems can be broadly classified into continuous and discrete architectures, while hybrid architectures combining these two approaches have also been investigated. In the continuous pinching-antenna architecture, the pinches can move continuously along the guiding structure, such that their radiation positions can be adjusted according to the communication requirements \cite{xu2025pin}. By contrast, the discrete pinching-antenna architecture adopts a practical implementation in which multiple pinches are either pre-installed at fixed locations along the waveguide or moved among predefined candidate locations as needed \cite{kaidi2025pin}. These different implementations provide different levels of flexibility and hardware complexity for realizing pinching-antenna systems in wireless networks.
\subsection{State-of-the-Art}
Building on the above classification, existing studies have investigated continuous and discrete pinching-antenna architectures from different perspectives. For continuous pinching-antenna architectures, early works focused on analytical modeling and performance characterization. In particular, the optimal antenna placement and its impact on system performance were analyzed to provide fundamental design insights \cite{ding2025analytical}. Building on this foundation, subsequent studies considered system-level designs by jointly optimizing antenna positions and transmit beamforming to improve communication performance \cite{zhou2026pin}. This framework was further extended to non-orthogonal multiple access (NOMA) scenarios, where antenna positions and power allocation were jointly optimized to enhance spectral efficiency \cite{xu2025pin2}. For discrete pinching-antenna architectures, discrete multi-waveguide pinching-antenna systems were explored, where each waveguide serves its associated users through joint waveguide assignment, antenna activation, and resource allocation \cite{kaidi2025pin2}. Under a different transmission structure, discrete pinching-antenna systems were also investigated in which all waveguides cooperatively serve all users through joint beamforming and antenna activation design \cite{xu2026pass}. Furthermore, the performance of discrete architectures was analytically characterized, revealing the impact of spatial discretization and quantifying the gap relative to ideal continuous placement \cite{tyrovolas2026pin}. In addition, several recent works have jointly examined continuous and discrete pinching-antenna architectures. In \cite{chen2025pin}, antenna placement and beamforming were optimized for multicast systems under both continuous positioning and discrete activation, illustrating the tradeoff between performance and implementation complexity. In \cite{wang2025pa}, a common signal model was adopted for both continuous and discrete activations, and the corresponding beamforming designs were analyzed under the two architectures. Despite the high flexibility provided by continuous architectures, their realization generally requires rapid and precise mechanical repositioning, which poses significant challenges for practical implementation. By contrast, discrete architectures are more suitable for practical deployment, where the design of activation mechanisms plays a central role in shaping both system realization and performance.
\subsection{Motivation and Contributions}
From an implementation perspective, most existing studies on pinching-antenna systems have focused on waveguide based realizations, where signals propagate along a guiding structure and are radiated through pinching elements. Such architectures are typically designed for high-frequency scenarios and rely on specialized waveguide structures, which may limit their applicability in more general wireless communication environments. In this context, leaky coaxial cables (LCXs) have recently been recognized as an alternative guiding medium for realizing pinching-antenna systems \cite{xu2025generalized}. In particular, by incorporating slot activation, the spatially distributed slots on LCX can be regarded as pre-installed pinching elements, thereby naturally corresponding to the discrete pinching-antenna architecture. Preliminary studies have explored this idea by establishing LCX based pinching-antenna models, where new channel characteristics such as angle dependent radiation patterns are introduced \cite{kaidi2025generalized}, and more flexible transmission designs such as dual-end feeding are enabled \cite{kaidi2026generalized}. Building upon these prior efforts, this paper further develops the LCX pinching-antenna system by introducing adjustable slot apertures, whereby conventional binary slot activation is generalized from an on/off scheme to continuous aperture adjustment. This allows each slot to be associated with a continuous weighting factor, which can be interpreted as a power coefficient, and hence enables more flexible slot-level radiation control.

The main contributions are summarized as follows:
\begin{itemize}[leftmargin=*]
\item An LCX pinching-antenna system with adjustable slot apertures is proposed, where conventional slot activation is generalized to aperture adjustment, resulting in a continuous weighting factor on the amplitude of each slot induced channel. Analytical results are derived to characterize the performance advantage of continuous aperture adjustment over binary slot activation and to reveal the effect of channel coherence on the achievable rate gain.
\item Based on the proposed framework, two time-division multiple access (TDMA) transmission schemes are developed to achieve a tradeoff between implementation complexity and transmission flexibility, namely static and dynamic aperture adjustment. The former employs a common aperture configuration for all users, whereas the latter adapts user-specific aperture coefficients across TDMA slots. For both schemes, the corresponding sum rate maximization problems are formulated.
\item Different optimization methods are developed to solve the resulting problems. For the static aperture adjustment problem, an iterative algorithm is constructed based on successive convex approximation (SCA) and quadratic transform. For the dynamic aperture adjustment problem, the original formulation is decomposed into user-specific subproblems, for which a low-complexity alternating optimization method with closed-form aperture updates is derived.
\item Simulation results demonstrate that the proposed LCX pinching-antenna system with aperture adjustment significantly outperforms both fixed-antenna and conventional LCX systems. Compared with binary slot activation, continuous aperture adjustment further enhances the sum rate and lowers the outage probability. The results also verify the analytical insights on the advantage of continuous aperture adjustment and the effect of channel coherence on the achievable gain.
\end{itemize}
\section{System Model}
\begin{figure}[!t]
\centering{\includegraphics[width=85mm]{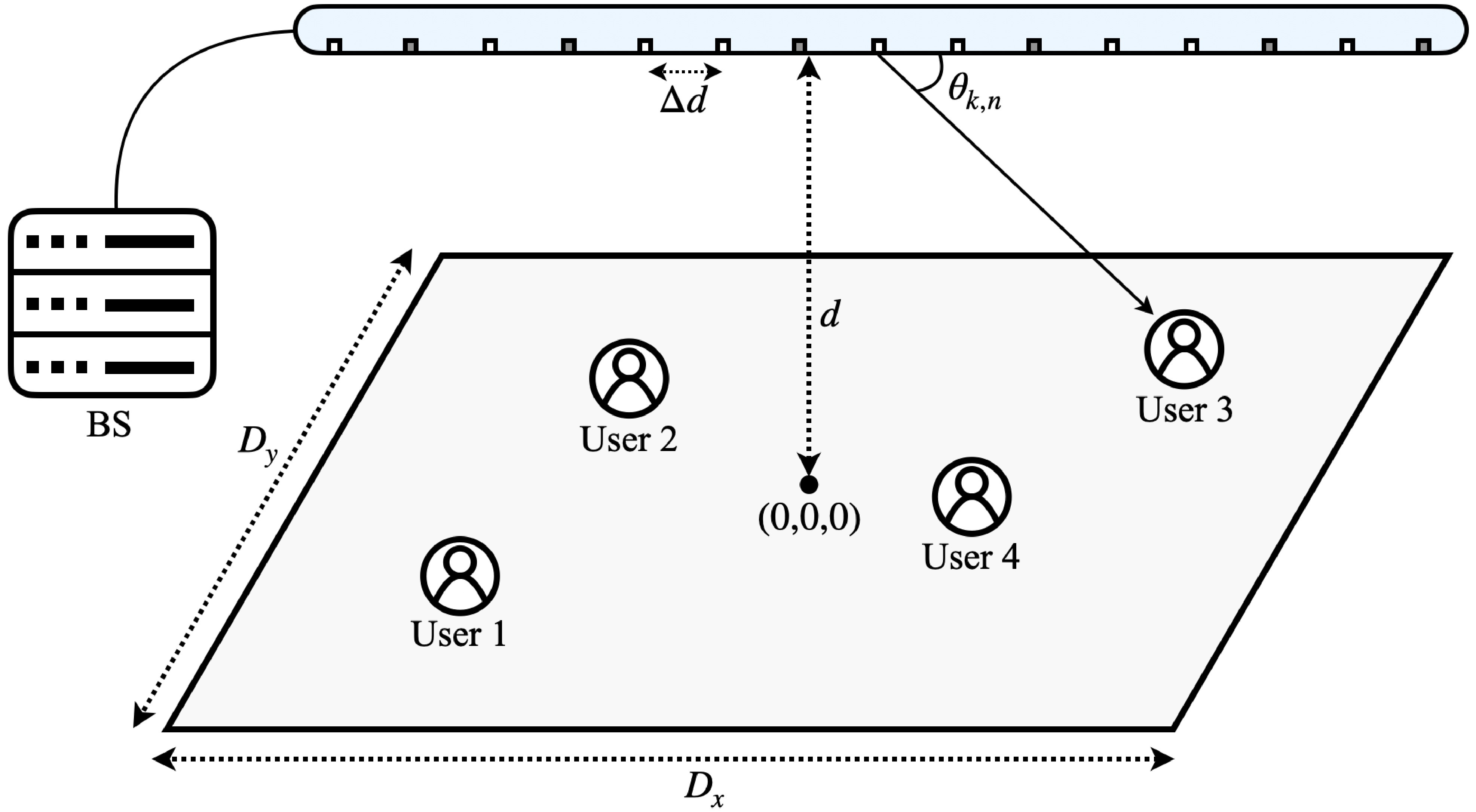}}
\caption{Illustration of the considered downlink LCX Pinching-Antenna System.}\vspace{-4mm}
\label{system}
\end{figure}

Consider a downlink LCX pinching-antenna system deployed within a rectangular service region of size $D_x \times D_y$, as shown in \fref{system}. The base station (BS) is connected to a single coaxial cable equipped with $K$ vertically oriented slots to serve $N$ users. The cable is installed at a height of $d$, and the spacing between two adjacent slots is denoted by $\Delta d$. The sets of users and slots are given by $\mathcal{N}=\{1,2,\dots,N\}$ and $\mathcal{K}=\{1,2,\dots,K\}$, respectively.
\subsection{Channel Model}
In the considered LCX pinching-antenna system, the signal injected at the feed port propagates along the coaxial cable and reaches each slot. Accounting for the amplitude attenuation and phase progression evaluated at the slot center \cite{torrance1996lcx}, the channel from the feed point located at $(-D_x/2,0,d)$ to the $k$-th slot can be expressed as follows:
\begin{equation}
h_k=10^{-\frac{\kappa_r}{20}(k-1)\Delta d}e^{-j\frac{2\pi}{\lambda}\sqrt{\varepsilon_r}(k-1)\Delta d},
\end{equation}
where $\kappa_r$ is the attenuation constant, $\lambda$ is the free-space wavelength, and $\varepsilon_r$ is the relative permittivity of the cable.

At each slot, the guided signal can leak into free space through the slot aperture. As the slots are oriented downward, the radiated field is modeled by a conical radiation pattern that depends on both the propagation distance and the elevation angle \cite{morgan1999lcx, yin2024lcx}. For the direct link from slot $k$ to user $n$, the LoS component is given by
\begin{equation}
h_{k,n}^\mathrm{LoS}=\eta\frac{e^{-j\frac{2\pi}{\lambda}\left\|\boldsymbol{\psi}_n-\boldsymbol{\psi}_k^\mathrm{slot}\right\|}}{\left\|\boldsymbol{\psi}_n-\boldsymbol{\psi}_k^\mathrm{slot}\right\|}\sin\!\left(\theta_{k,n}\right),
\end{equation}
where $\eta=\frac{\lambda}{4\pi}$ is the free-space path-loss normalization factor, $\boldsymbol{\psi}_n=(x_n,y_n,0)$ is the position of user $n$, and $\boldsymbol{\psi}_k^\mathrm{slot}=(-D_x/2+(k-1)\Delta d,0,d)$ is the location of slot $k$. Moreover, $\|\boldsymbol{\psi}_n-\boldsymbol{\psi}_k^\mathrm{slot}\|$ is the distance between user $n$ and slot $k$, and $\theta_{k,n}$ is the corresponding elevation angle, satisfying $\sin(\theta_{k,n})=\frac{d}{\|\boldsymbol{\psi}_n-\boldsymbol{\psi}_k^\mathrm{slot}\|}$, as illustrated in \fref{system}.

To model the non-line-of-sight (NLoS) component, $L$ scatterers are introduced to characterize the reflected propagation paths, which play an important role in general wireless environments, especially in lower-frequency scenarios. Specifically, the radiated signal from a slot first propagates to the scatterers and is then reflected toward the user. The resulting NLoS channel between slot $k$ and user $n$ is given by
\begin{equation}
h_{k,n}^\mathrm{NLoS}\!=\!\eta\!\sum_{\ell=1}^{L}\!\frac{\delta_\ell e^{-j\frac{2\pi}{\lambda}\left(\left\|\boldsymbol{\psi}_\ell^\mathrm{scat}\!-\boldsymbol{\psi}_k^\mathrm{slot}\right\|+\left\|\boldsymbol{\psi}_n\!-\boldsymbol{\psi}_\ell^\mathrm{scat}\right\|\right)}}{\left\|\boldsymbol{\psi}_\ell^\mathrm{scat}-\boldsymbol{\psi}_k^\mathrm{slot}\right\|\left\|\boldsymbol{\psi}_n-\boldsymbol{\psi}_\ell^\mathrm{scat}\right\|}\!\sin\!\left(\theta_{k,\ell}\right),
\end{equation}
where $\delta_\ell$ is the complex gain associated with the $\ell$-th scattering path, $\boldsymbol{\psi}_\ell^\mathrm{scat}=(x_\ell,y_\ell,z_\ell)$ is the location of scatterer $\ell$, and $\theta_{k,\ell}$ is the elevation angle from slot $k$ to scatterer $\ell$.

By combining the guided wave propagation along the cable with the wireless propagation from slots to users, the channel from the feed point to user $n$ via slot $k$ can be expressed as
\begin{equation}
h_{k,n}=h_k\!\left(h_{k,n}^\mathrm{LoS}+h_{k,n}^\mathrm{NLoS}\right).
\end{equation}
\subsection{Slot Aperture Adjustment Design}
\begin{figure}[!t]
\centering{\includegraphics[width=64mm]{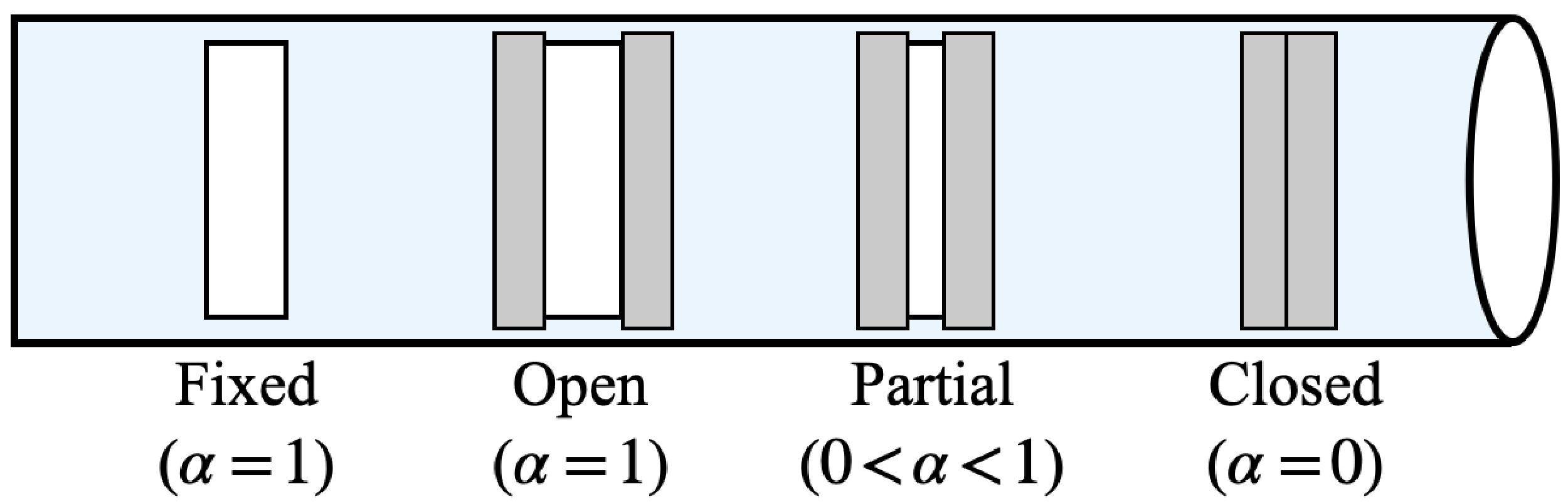}}
\caption{An illustration of the symmetric sliding shutter based slot aperture adjustment mechanism.}\vspace{-4mm}
\label{lcxaa}
\end{figure}

In existing LCX related works, slot-level radiation control has been investigated from a design and modeling perspective. Early LCX designs have explored structural reconfiguration to enable adjustable coupling loss and radiation characteristics, such as through periodic metallic patch loading along the cable \cite{wang2001lcx1}. In more recent LCX related studies, the radiation state of each slot has been modeled via a binary slot activation variable, where a slot is either radiating (on) or blocked (off) \cite{nagayama2022lcx1, kaidi2025generalized}. From an implementation perspective, various practical hardware realizations supporting such on/off operation have been disclosed in a related patent \cite{asplund2016leaky}. Building on the existing studies of binary slot activation, this paper further proposes a slot aperture adjustment mechanism, in which the effective radiating aperture of each slot can be tuned to control the leakage strength. Prior studies have shown that the slot size can significantly affect the coupling loss and the radiated power \cite{wang2001lcx, siddiqui2020lcx}. In particular, enlarging the slot opening reduces the coupling loss while enhancing the radiated power into free space. As shown in \fref{lcxaa}, this mechanism can be realized by a symmetric two-blade shutter that closes toward the aperture center, thereby maintaining a fixed geometric radiation center throughout the adjustment process\footnote{The aperture adjustment mechanism is introduced as a modeling oriented hardware abstraction rather than a hardware prototype. Practical implementation may rely on reconfigurable mechanical or electromechanical structures, with reduced complexity designs adopting grouped control over multiple slots instead of fully independent per-slot actuation.}. As a result, the shutter operation modulates the radiation amplitude without altering the radiation direction or introducing additional phase shifts\footnote{Alternatively, an asymmetric single-blade shutter can be adopted, such that partial aperture blocking leads to a displacement of the slot's effective phase center. In this case, when $0<\alpha_k<1$, the effective in-cable propagation distance and the elevation angle toward the user are both altered, resulting in a controllable phase shift in addition to amplitude scaling.}.

To model the shutter controlled radiation characteristics, a slot aperture adjustment coefficient $\alpha_k \in [0,1]$ is introduced, where $\alpha_k=0$ corresponds to a fully blocked slot and $\alpha_k=1$ corresponds to a fully open aperture, as shown in \fref{lcxaa}. In the considered model, the aperture adjustment coefficient characterizes the effect of slot opening on the leakage strength, and hence scales the leaked wave amplitude and the resulting channel magnitude. The radiation pattern variation induced by aperture adjustment is assumed negligible relative to the dominant amplitude modulation effect. Accordingly, the effective channel of user $n$ is given by
\begin{equation}
h_n=\sum_{k=1}^K \alpha_k h_{k,n}=\mathbf{h}_n^\mathsf{T}\boldsymbol{\alpha},
\end{equation}
where $\mathbf{h}_n=[h_{1,n}, \dots, h_{K,n}]^\mathsf{T}$ and $\boldsymbol{\alpha}=[\alpha_1, \dots, \alpha_K]^\mathsf{T}$ are the slot induced channel vector and the aperture adjustment vector, respectively.

The following remarks clarify the interpretation and generality of the proposed model.
\begin{remark}
The aperture adjustment coefficient $\alpha_k$ jointly represents slot availability and effective aperture control, where $\alpha_k=0$ corresponds to a blocked slot and $\alpha_k\in(0,1]$ scales the radiated field intensity through fractional aperture opening, thereby serving as an amplitude weighting factor applied to the slot induced channel. This unified model captures both binary activation and continuous radiation control within a compact channel formulation.
\end{remark}

\begin{remark}
The continuous aperture adjustment model generalizes the conventional binary slot activation scheme by allowing fractional aperture control. While the binary model restricts each slot to be either fully blocked or fully open, the continuous formulation enables fine-grained amplitude control, thereby enlarging the design space and providing additional flexibility for performance optimization.
\end{remark}

To characterize the advantage of the continuous aperture adjustment model, the following proposition is provided.
\begin{proposition}
Define the feasible sets of binary slot activation and continuous aperture adjustment as
\begin{equation}
\mathcal{A}_\mathrm{bin}\triangleq\{\boldsymbol{\alpha}\in\{0,1\}^K\},
\end{equation}
and
\begin{equation}
\mathcal{A}_\mathrm{con}\triangleq\{\boldsymbol{\alpha}\in[0,1]^K\}.
\end{equation}
Then $\mathcal{A}_\mathrm{bin}\subseteq\mathcal{A}_\mathrm{con}$. For any objective function $f(\boldsymbol{\alpha})$, it holds that
\begin{equation}
f_\mathrm{con}^\star=\max_{\boldsymbol{\alpha}\in\mathcal{A}_\mathrm{con}}f(\boldsymbol{\alpha})\ge\max_{\boldsymbol{\alpha}\in\mathcal{A}_\mathrm{bin}}f(\boldsymbol{\alpha})=f_\mathrm{bin}^\star.
\end{equation}
The same conclusion holds for minimization problems, with the inequality direction reversed.
\end{proposition}
\begin{IEEEproof}
According to the definition, the feasible set of the binary scheme satisfies $\mathcal{A}_\mathrm{bin}=\{0,1\}^K\subseteq [0,1]^K=\mathcal{A}_\mathrm{con}$. Therefore, any binary slot activation vector $\boldsymbol{\alpha}\in\mathcal{A}_\mathrm{bin}$ is also feasible for the continuous aperture adjustment problem. Since the objective function $f(\boldsymbol{\alpha})$ is maximized over a superset of feasible solutions in the continuous case, the corresponding optimal value cannot be smaller than that obtained over the subset $\mathcal{A}_{\mathrm{bin}}$, i.e.,
\begin{equation}
\max_{\boldsymbol{\alpha}\in\mathcal{A}_\mathrm{con}} f(\boldsymbol{\alpha})\ge\max_{\boldsymbol{\alpha}\in\mathcal{A}_\mathrm{bin}} f(\boldsymbol{\alpha}).
\end{equation}
This completes the proof.
\end{IEEEproof}
\subsection{Performance Analysis}
In this subsection, the performance of the proposed aperture adjustment model is analyzed for the single-user case. The radiation characteristics are captured by the aperture adjustment coefficient $\alpha_k$, which scales the leaked wave strength and hence the channel amplitude. Since the radiated field amplitude varies linearly with $\alpha_k$, the corresponding radiated power fraction from slot $k$ is proportional to $\alpha_k^2$. The received signal at user $n$ is given by
\begin{equation}
y_n=\sqrt{\frac{P_t}{A_\mathrm{eff}}}h_ns_n+n_0=\sqrt{\frac{P_t}{A_\mathrm{eff}}}\sum_{k=1}^K \alpha_k h_{k,n}s_n+n_0,
\end{equation}
where $P_t$ is the total transmit power, $A_\mathrm{eff}=\sum_{k=1}^{K}\alpha_k^2=\|\boldsymbol{\alpha}\|_2^2$ is the effective radiating aperture normalization factor, $s_n$ is the desired signal of user $n$, and $n_0$ is the additive white Gaussian noise. Specifically, $A_\mathrm{eff}$ is introduced to normalize the transmit power in accordance with the aperture adjustment coefficients, enabling a fair comparison among different aperture adjustment patterns under an identical transmit power budget. To ensure system operability, the case where all slots are deactivated is excluded by imposing $A_\mathrm{eff}\!>\!0$. The achievable data rate of user $n$ can be presented as follows:
\begin{equation}
R_n(\boldsymbol{\alpha})=\log_2\!\left(1\!+\!\frac{P_t|h_n|^2}{\sigma^2A_\mathrm{eff}}\right)=\log_2\!\left(1\!+\!\frac{P_t|\mathbf{h}_n^\mathsf{T}\boldsymbol{\alpha}|^2}{\sigma^2\|\boldsymbol{\alpha}\|_2^2}\right),
\end{equation}
where $\sigma^2$ is the noise power.

To derive conditions under which aperture adjustment can strictly outperform a given binary slot activation pattern, the achievable data rate is adopted as the objective, leading to the following proposition.
\begin{proposition}\label{strict}
For any binary activation vector $\boldsymbol{\alpha}_{\mathrm{bin}}\in\mathcal{A}_{\mathrm{bin}}$ with the active slot index set $\mathcal{S}\triangleq\{k\in\mathcal{K}|\alpha_{\mathrm{bin},k}=1\}$, where $|\mathcal{S}|\ge 1$. If either of the following inequalities holds:
\begin{subequations}\label{strictcond}
\begin{empheq}[left=\empheqlbrace]{align}
\Re\!\left\{(\mathbf{h}_n^\mathsf{T}\boldsymbol{\alpha}_\mathrm{bin})h_{k,n}^*\right\}&<\frac{|\mathbf{h}_n^\mathsf{T}\boldsymbol{\alpha}_\mathrm{bin}|^2}{\|\boldsymbol{\alpha}_\mathrm{bin}\|_2^2}, &\exists\,k\in\mathcal{S},\\
\Re\!\left\{(\mathbf{h}_n^\mathsf{T}\boldsymbol{\alpha}_\mathrm{bin})h_{k,n}^*\right\}&>0, &\exists\,k\notin\mathcal{S},
\end{empheq}
\end{subequations}
there exists a feasible $\boldsymbol{\alpha}\!\in\!\mathcal{A}_\mathrm{con}$ such that $R_n(\boldsymbol{\alpha})\!>\!R_n(\boldsymbol{\alpha}_\mathrm{bin})$.
\end{proposition}
\begin{IEEEproof}
Refer to Appendix~A.
\end{IEEEproof}

The following remark provides further insight into the conditions stated in Proposition~\ref{strict}.
\begin{remark}
The term $\Re\{(\mathbf{h}_n^\mathsf{T}\boldsymbol{\alpha}_\mathrm{bin})h_{k,n}^*\}$ measures the in-phase projection of the $k$-th slot channel onto the composite channel $\mathbf{h}_n^\mathsf{T}\boldsymbol{\alpha}_\mathrm{bin}$, thereby capturing both phase alignment and channel magnitude effects. Condition (\ref{strictcond}a) indicates that if an active slot has a projection smaller than $|\mathbf{h}_n^\mathsf{T}\boldsymbol{\alpha}_\mathrm{bin}|^2/\|\boldsymbol{\alpha}_\mathrm{bin}\|_2^2$, then partially closing its aperture increases the normalized SNR. On the other hand, condition (\ref{strictcond}b) indicates that if an inactive slot has a positive projection, then partially opening its aperture increases the normalized SNR.
\end{remark}

The following proposition characterizes an upper bound on the achievable rate gain over the conventional LCX baseline in the high-SNR regime.
\begin{proposition}\label{upperbound}
Consider the conventional LCX baseline corresponding to the fully active binary pattern $\boldsymbol{\alpha}_{\mathrm{bin}}=\mathbf{1}$. Under aperture adjustment, the coefficient vector $\boldsymbol{\alpha}$ is designed over all slots subject to $\alpha_k\in[0,1]$ for all $k\in\mathcal{K}$ and $\boldsymbol{\alpha}\neq\mathbf{0}$. In the high-SNR regime, the achievable rate gain over the conventional LCX baseline, denoted by $\Delta R_n$, can be upper bounded as follows:
\begin{equation}
0 \le\Delta R_n\lesssim-\log_2\!\left(\frac{\left|\sum_{k\in\mathcal{K}} h_{k,n}\right|^2}{K\!\sum_{k\in\mathcal{K}}|h_{k,n}|^2}\right).
\end{equation}
\end{proposition}
\begin{IEEEproof}
Refer to Appendix~B.
\end{IEEEproof}

The implications of Proposition~\ref{upperbound} are further discussed in the following remarks.
\begin{remark}
The active slot index set $\mathcal{S}$ in Proposition~\ref{strict} can represent any binary activation pattern, including an optimal binary solution. In Proposition~\ref{upperbound}, the conventional LCX baseline with all slots active corresponds to the special case $\mathcal{S}=\mathcal{K}$. In both settings, aperture adjustment can provide further improvement through fractional aperture control.
\end{remark}

\begin{remark}
The high-SNR rate gain upper bound in Proposition~\ref{upperbound} is determined by the following coherence factor
\begin{equation}
\rho(\mathcal{K})\triangleq\frac{\left|\sum_{k\in\mathcal{K}}h_{k,n}\right|^2}{K\!\sum_{k\in\mathcal{K}}|h_{k,n}|^2},
\end{equation}
where $0<\rho(\mathcal{K})\le 1$. This factor quantifies the degree of constructive combining among the slot induced channels under binary activation. Specifically, $\rho(\mathcal{K})\approx 1$ indicates near phase alignment and negligible gain, whereas $\rho(\mathcal{K})\ll 1$ corresponds to severe phase dispersion, under which amplitude re-weighting enabled by aperture adjustment can provide a larger improvement.
\end{remark}

\begin{remark}
By writing $h_{k,n}=|h_{k,n}|e^{j\phi_{k,n}}$, it follows that
\begin{equation}
\left|\sum_{k\in\mathcal{K}}\!h_{k,n}\right|^2\!\!=\!\sum_{k\in\mathcal{K}}\!|h_{k,n}|^2\!+\!2\!\!\sum_{\substack{i<j\\i,j\in\mathcal{K}}}\!\!|h_{i,n}||h_{j,n}|\cos(\phi_{i,n}\!-\!\phi_{j,n}),
\end{equation}
which implies that $\rho(\mathcal{K})$ decreases with phase dispersion via the cosine cross terms and is further influenced by amplitude imbalance across slots.
\end{remark}

In the considered LCX pinching-antenna system, the channel phase $\phi_{k,n}$ typically varies rapidly with the slot index $k$ due to the combined effects of guided wave propagation along the cable and the subsequent free-space radiation paths. As a result, the cross terms $\cos(\phi_{i,n}-\phi_{j,n})$ in the coherent summation are often negative, leading to pronounced destructive combining across slots. This reduces $|\sum_{k\in\mathcal{K}}h_{k,n}|^2$, thereby enlarging the performance gap between the full activation binary baseline and aperture adjustment based amplitude re-weighting.
\subsection{Signal Model}
Based on the LCX pinching-antenna channel model with adjustable slot apertures, a TDMA transmission framework is considered, providing two different aperture adjustment strategies.
\subsubsection{TDMA with Static Aperture Adjustment}
In the TDMA framework with static aperture adjustment, the slot aperture configuration is determined prior to transmission and remains fixed across all users' time slots. Specifically, a common aperture adjustment vector is designed to balance the overall channel conditions of the served users, without adapting to individual user channels in each TDMA interval. During transmission, users are served sequentially in orthogonal time slots while the radiation distribution of the LCX remains unchanged. This strategy reduces implementation complexity by avoiding real-time reconfiguration of slot apertures and enables a unified radiation structure that provides stable coverage across the service region. The received signal for user $n$ is given by
\begin{equation}
y_n^\mathrm{sta}=\!\sqrt{\frac{P_t}{A_\mathrm{eff}}}h_ns_n\!+\!n_0=\!\sqrt{\frac{P_t}{A_\mathrm{eff}}}\!\sum_{k=1}^K\!\alpha_k h_{k,n}s_n\!+\!n_0.
\end{equation}
The achievable data rate of user $n$ can be presented as follows:
\begin{equation}
R_n^\mathrm{sta}=\frac{1}{N}\log_2\!\left(1\!+\!\frac{P_t|h_n|^2}{\sigma^2A_\mathrm{eff}}\right)\!=\!\frac{1}{N}\log_2\!\left(1\!+\!\frac{P_t|\mathbf{h}_n^\mathsf{T}\boldsymbol{\alpha}|^2}{\sigma^2\|\boldsymbol{\alpha}\|_2^2}\right).
\end{equation}
\subsubsection{TDMA with Dynamic Aperture Adjustment}
In TDMA with dynamic aperture adjustment, the slot aperture pattern is designed on a per-user basis and updated across TDMA slots. Specifically, when user $n$ is scheduled, a dedicated aperture adjustment vector is applied for that user's transmission interval and can be reconfigured for subsequent users. Accordingly, an additional subscript is introduced for the activation coefficient, i.e., $\alpha_{k,n}$, to denote the activation state of slot $k$ during the transmission interval of user $n$. This user-specific reconfiguration enables customized amplitude weighted radiation patterns that better match each user's channel, at the cost of increased control and switching overhead. The received signal for user $n$ at the corresponding time slot is given by
\begin{equation}
y_n^\mathrm{dyn}=\!\sqrt{\frac{P_t}{A_{\mathrm{eff},n}}}h_ns_n\!+\!n_0=\!\sqrt{\frac{P_t}{A_{\mathrm{eff},n}}}\!\sum_{k=1}^K\!\alpha_{k,n} h_{k,n}s_n\!+\!n_0.
\end{equation}
Hence, the achievable data rate of user $n$ is
\begin{equation}
R_n^\mathrm{dyn}=\frac{1}{N}\log_2\!\left(1\!+\!\frac{P_t|h_n|^2}{\sigma^2A_{\mathrm{eff},n}}\right)\!=\!\frac{1}{N}\log_2\!\left(1\!+\!\frac{P_t|\mathbf{h}_n^\mathsf{T}\boldsymbol{\alpha}_n|^2}{\sigma^2\|\boldsymbol{\alpha}_n\|_2^2}\right),
\end{equation}
where $\boldsymbol{\alpha}_n$ is the aperture adjustment vector associated with the transmission interval of user $n$.
\section{Problem Formulation}
In the proposed LCX pinching-antenna system, the slot aperture adjustment coefficients provide a flexible mechanism to shape the radiated field distribution, thereby improving signal combining at the users. This section considers the optimization of the aperture coefficients to maximize the overall system throughput subject to quality-of-service (QoS) constraints. The corresponding optimization problem for the static aperture adjustment is formulated as follows.
\begin{subequations}
\begin{empheq}{align}
\max_{\boldsymbol{\alpha}}\quad & \sum_{n=1}^NR_n^\mathrm{sta}\\
\textrm{s.t.} \quad & \mathbf{0}\preceq \boldsymbol{\alpha}\preceq \mathbf{1},\\
& \|\boldsymbol{\alpha}\|_2^2>0,\\
& R_n^\mathrm{sta} \ge R_\mathrm{min},\ \forall n\in\mathcal{N}.
\end{empheq}
\label{staproblem}
\end{subequations}\vspace{-2mm}\\
In problem \eqr{staproblem}, constraint (\ref{staproblem}b) indicates the feasible range of the aperture adjustment coefficients, and constraint (\ref{staproblem}c) ensures that at least one slot remains active to maintain system operability. In constraint (\ref{staproblem}d), a minimum data rate $R_\mathrm{min}$ is guaranteed for each user.

Similarly, the optimization problem under dynamic aperture adjustment is formulated as follows:
\begin{subequations}
\begin{empheq}{align}
\max_{\{\boldsymbol{\alpha}_n\}_{n\in\mathcal{N}}}\quad & \sum_{n=1}^NR_n^\mathrm{dyn}\\
\textrm{s.t.} \quad & \mathbf{0}\preceq \boldsymbol{\alpha}_n \preceq \mathbf{1},\ \forall n\in\mathcal{N},\\
& \|\boldsymbol{\alpha}_n\|_2^2>0,\ \forall n\in\mathcal{N},\\
& R_n^\mathrm{dyn} \ge R_\mathrm{min},\ \forall n\in\mathcal{N},
\end{empheq}
\label{dynproblem}
\end{subequations}\vspace{-2mm}\\
where $\{\boldsymbol{\alpha}_n\}_{n\in\mathcal{N}}$ is the collection of user-specific aperture vectors. The constraints remain identical in interpretation to those in the static case, with user-specific aperture coefficients.
\section{Solution for Static Aperture Adjustment}
By defining $\rho\triangleq\frac{P_t}{\sigma^2}$, the achievable data rate of user $n$ under static aperture adjustment can be rewritten as follows:
\begin{equation}
R_n^\mathrm{sta}=\frac{1}{N}\log_2(1+\rho f_n(\boldsymbol{\alpha})),
\end{equation}
where
\begin{equation}
f_n(\boldsymbol{\alpha})=\frac{|\mathbf{h}_n^\mathsf{T}\boldsymbol{\alpha}|^2}{\|\boldsymbol{\alpha}\|_2^2}.
\end{equation}
Moreover, the QoS constraint in (\ref{staproblem}d) can be rewritten as $f_n(\boldsymbol{\alpha})\ge \gamma_\mathrm{min}$ with
\begin{equation}\label{gamma}
\gamma_\mathrm{min}\triangleq\frac{1}{\rho}\!\left(2^{NR_\mathrm{min}}-1\right).
\end{equation}
The objective contains $\log_2(1+\rho f_n(\boldsymbol{\alpha}))$, which is concave with $f_n(\boldsymbol{\alpha})\ge 0$. Since the factor $1/N$ is a positive constant, it is omitted in the following derivations. To enable a tractable iterative algorithm, a first-order lower bound of the concave function is applied at the $t$-th iteration point $f_n(\boldsymbol{\alpha}^{(t)})$ as follows:
\begin{align}
\log_2(1\!+\!\rho f_n(\boldsymbol{\alpha}))\ge &\log_2(1\!+\!\rho f_n(\boldsymbol{\alpha}^{(t)}))\\\nonumber
&+\!\frac{\rho}{\ln(2)(1\!+\!\rho f_n(\boldsymbol{\alpha}^{(t)}))}(f_n(\boldsymbol{\alpha})\!-\!\!f_n(\boldsymbol{\alpha}^{(t)})).
\end{align}
By defining a weight as follows:
\begin{equation}\label{weight1}
\omega_n^{(t)}\triangleq\frac{\rho}{\ln(2)(1+\rho f_n(\boldsymbol{\alpha}^{(t)}))},
\end{equation}
the objective function at the $t$-th iteration is given by
\begin{equation}
\sum_{n=1}^N\log_2(1\!+\!\rho f_n(\boldsymbol{\alpha}^{(t)}))\!+\!\sum_{n=1}^N\omega_n^{(t)}(f_n(\boldsymbol{\alpha})\!-\!\!f_n(\boldsymbol{\alpha}^{(t)})).
\end{equation}
By removing the constants, the following problem is solved at the $t$-th iteration:
\begin{subequations}
\begin{empheq}{align}
\max_{\boldsymbol{\alpha}}\quad & \sum_{n=1}^N\omega_n^{(t)}f_n(\boldsymbol{\alpha})\\
\textrm{s.t.} \quad & f_n(\boldsymbol{\alpha}) \ge \gamma_\mathrm{min},\ \forall n\in\mathcal{N},\\\nonumber
& \text{(\ref{staproblem}b)}, \text{and (\ref{staproblem}c)}.
\end{empheq}
\label{staproblem1}
\end{subequations}\vspace{-2mm}

By introducing an auxiliary variable $\beta_n$ for each user $n$ and applying the quadratic transform in \cite{shen2018fp}, the following expression can be obtained:
\begin{equation}\label{qt1}
f_n(\boldsymbol{\alpha})=\max_{\beta_n\in\mathbb{C}}\ g_n(\boldsymbol{\alpha},\beta_n),
\end{equation}
where
\begin{equation}
g_n(\boldsymbol{\alpha},\beta_n)\triangleq 2\Re\{\beta_n^*\mathbf{h}_n^\mathsf{T}\boldsymbol{\alpha}\}-|\beta_n|^2\|\boldsymbol{\alpha}\|_2^2.
\end{equation}
Note that \eqr{qt1} indicates that
\begin{equation}
f_n(\boldsymbol{\alpha}) \ge g_n(\boldsymbol{\alpha},\beta_n).
\end{equation}
Moreover, for any given $\boldsymbol{\alpha}\neq\mathbf{0}$, the optimal value of $\beta_n$ can be obtained in closed form. Accordingly, at the $t$-th iteration, the auxiliary variable is updated as
\begin{equation}\label{beta1}
\beta_n^{(t)}=\frac{\mathbf{h}_n^\mathsf T\boldsymbol{\alpha}^{(t)}}{\|\boldsymbol{\alpha}^{(t)}\|_2^2}.
\end{equation}
By fixing $\{\beta_n^{(t)}\}_{n\in\mathcal{N}}$, problem \eqr{staproblem1} can be approximated by the following convex problem:
\begin{subequations}
\begin{empheq}{align}
\max_{\boldsymbol{\alpha}}\quad & \sum_{n=1}^N\omega_n^{(t)} g_n(\boldsymbol{\alpha},\beta_n^{(t)})\\
\textrm{s.t.}\quad & g_n(\boldsymbol{\alpha},\beta_n^{(t)}) \ge \gamma_\mathrm{min},\ \forall n\in\mathcal{N},\\\nonumber
& \text{(\ref{staproblem}b)}, \text{and (\ref{staproblem}c)}.
\end{empheq}
\label{staproblem2}
\end{subequations}\vspace{-2mm}\\
Since $f_n(\boldsymbol{\alpha}) \ge g_n(\boldsymbol{\alpha},\beta_n^{(t)})$, the QoS constraint in (\ref{staproblem1}b) can be guaranteed by (\ref{staproblem2}b). Problem \eqr{staproblem2} corresponds to the problem solved at the $t$-th iteration, where $\{\omega_n^{(t)}\}$ and $\{\beta_n^{(t)}\}$ are fixed coefficients updated based on the previous iterate  $\boldsymbol{\alpha}^{(t)}$. 

For fixed $\{\beta_n^{(t)}\}$, function $g_n(\boldsymbol{\alpha},\beta_n^{(t)})$ is concave quadratic in $\boldsymbol{\alpha}$, since it consists of an affine term $2\Re\{\beta_n^{(t)*}\mathbf{h}_n^\mathsf{T}\boldsymbol{\alpha}\}$ and a negative quadratic term $-|\beta_n^{(t)}|^2\|\boldsymbol{\alpha}\|_2^2$. Therefore, the objective in \eqr{staproblem2} is concave in $\boldsymbol{\alpha}$. Moreover, the constraint $g_n(\boldsymbol{\alpha},\beta_n^{(t)})\ge \gamma_{\min}$ defines a convex set since it is the superlevel set of a concave function. In addition, the remaining constraints in (\ref{staproblem}b) and (\ref{staproblem}c) are convex. Hence, \eqr{staproblem2} is a convex optimization problem and can be efficiently solved by standard solvers, e.g., CVX.

The overall iterative procedure is summarized as follows. Given a feasible initialization $\boldsymbol{\alpha}^{(0)}$ and the iteration index $t\in\{0,1,\dots\}$, the following quantities are first computed for all $n\in\mathcal{N}$:
\begin{equation}
f_n(\boldsymbol{\alpha}^{(t)})=\frac{|\mathbf{h}_n^\mathsf{T}\boldsymbol{\alpha}^{(t)}|^2}{\|\boldsymbol{\alpha}^{(t)}\|_2^2}.
\end{equation}
Then, the weight and auxiliary variables are updated according to \eqr{weight1} and \eqr{beta1}. Afterwards, the convex problem in \eqr{staproblem2} is solved to obtain the updated aperture vector $\boldsymbol{\alpha}^{(t+1)}$. This procedure is repeated until convergence, e.g., when the increment of the original objective value $\sum_{n=1}^N R_n^\mathrm{sta}$ falls below a prescribed threshold. During this procedure, constraint (\ref{staproblem}c) is implemented as $\|\boldsymbol{\alpha}\|_2^2\ge \varepsilon$ with a small constant $\varepsilon>0$ to exclude the trivial all-zero solution and avoid numerical issues.
\section{Solution for Dynamic Aperture Adjustment}
Under the TDMA framework with dynamic aperture adjustment, each user is served in an orthogonal time slot with a dedicated aperture vector. Therefore, the aperture coefficients used for different users do not couple with each other. In particular, both the objective function and the constraints in \eqr{dynproblem} are separable with respect to $\{\boldsymbol{\alpha}_n\}_{n\in\mathcal{N}}$. As a result, problem \eqr{dynproblem} is decomposed into $N$ independent subproblems. For user $n$, the corresponding subproblem is given by
\begin{subequations}
\begin{empheq}{align}
\max_{\boldsymbol{\alpha}_n}\quad & R_n^\mathrm{dyn}\\
\textrm{s.t.} \quad & \mathbf{0}\preceq \boldsymbol{\alpha}_n \preceq \mathbf{1},\\
& \|\boldsymbol{\alpha}_n\|_2^2>0,\\
& R_n^\mathrm{dyn} \ge R_\mathrm{min}.
\end{empheq}
\label{dynproblem1}
\end{subequations}\vspace{-2mm}\\
Since $\log_2(1+x)$ is monotonically increasing in $x$, maximizing $R_n^\mathrm{dyn}$ is equivalent to maximizing the following term:
\begin{equation}
f_n(\boldsymbol{\alpha}_n)\triangleq\frac{|\mathbf{h}_n^\mathsf{T}\boldsymbol{\alpha}_n|^2}{\|\boldsymbol{\alpha}_n\|_2^2}.
\end{equation}
Therefore, problem \eqr{dynproblem1} is equivalent to the following problem:
\begin{subequations}
\begin{empheq}{align}
\max_{\boldsymbol{\alpha}_n}\quad & f_n(\boldsymbol{\alpha}_n)\\
\textrm{s.t.} \quad & f_n(\boldsymbol{\alpha}_n) \ge \gamma_\mathrm{min},\\\nonumber
& \text{(\ref{dynproblem1}b)}, \text{and (\ref{dynproblem1}c)},
\end{empheq}
\label{dynproblem2}
\end{subequations}\vspace{-2mm}\\
where $\gamma_\mathrm{min}$ follows the definition in the static case as \eqr{gamma}.

By introducing an auxiliary variable $\beta$, quadratic transform can be utilized to obtain the following equation:
\begin{equation}
f_n(\boldsymbol{\alpha}_n)=\max_{\beta\in\mathbb{C}} g(\boldsymbol{\alpha}_n,\beta),
\end{equation}
where
\begin{equation}
g(\boldsymbol{\alpha}_n,\beta)=2\Re\{\beta^*\mathbf{h}_n^\mathsf{T}\boldsymbol{\alpha}_n\}-|\beta|^2\|\boldsymbol{\alpha}_n\|_2^2.
\end{equation}
Given $\boldsymbol{\alpha}_n^{(t)}\neq\mathbf{0}$, the optimality condition of the quadratic transform leads to the auxiliary variable update as follows:
\begin{equation}\label{beta2}
\beta^{(t)}=\frac{\mathbf{h}_n^\mathsf{T}\boldsymbol{\alpha}_n^{(t)}}{\|\boldsymbol{\alpha}_n^{(t)}\|_2^2}.
\end{equation}
With the fixed auxiliary variable at the $t$-th iteration, $\boldsymbol{\alpha}_n$ can be updated by solving the following problem:
\begin{subequations}
\begin{empheq}{align}
\max_{\boldsymbol{\alpha}_n}\quad & g(\boldsymbol{\alpha}_n,\beta^{(t)})\\
\textrm{s.t.}\quad & g(\boldsymbol{\alpha}_n,\beta^{(t)}) \ge \gamma_\mathrm{min},\\\nonumber
& \text{(\ref{dynproblem1}b)}, \text{and (\ref{dynproblem1}c)}.
\end{empheq}
\label{dynproblem3}
\end{subequations}\vspace{-2mm}

Unlike the static design where a common aperture vector jointly affects all users, the dynamic formulation considers user-specific aperture vectors in orthogonal TDMA slots. As a result, each subproblem aims at maximizing the data rate of the corresponding user. In this case, the QoS constraint $g(\boldsymbol{\alpha}_n,\beta^{(t)})\ge\gamma_\mathrm{min}$ serves only as a feasibility condition. If the maximum achievable value of $g(\boldsymbol{\alpha}_n,\beta^{(t)})$ subject to the box constraints is below $\gamma_\mathrm{min}$, then the target rate $R_\mathrm{min}$ is infeasible for user $n$ under the current system parameters. Therefore, the QoS constraint can be omitted in the optimization step and verified after convergence. Accordingly, problem \eqr{dynproblem3} is rewritten as follows:
\begin{subequations}
\begin{empheq}{align}
\max_{\boldsymbol{\alpha}_n}\quad & g(\boldsymbol{\alpha}_n,\beta^{(t)})\\
\textrm{s.t.}\quad & \mathbf{0}\preceq \boldsymbol{\alpha}_n \preceq \mathbf{1}.
\end{empheq}
\label{dynproblem4}
\end{subequations}\vspace{-2mm}\\
Note that constraint (\ref{dynproblem1}c) only excludes the trivial all-zero solution to ensure a strictly positive denominator in the auxiliary update. Since $\boldsymbol{\alpha}_n=\mathbf{0}$ gives $g(\mathbf{0},\beta^{(t)})=0$ and is not optimal for \eqr{dynproblem4} except for degenerate cases where $\mathbf{a}_n^{(t)}\preceq\mathbf{0}$, omitting (\ref{dynproblem1}c) does not change the optimal solution. A nonzero initialization $\boldsymbol{\alpha}_n^{(0)}\neq\mathbf{0}$ is therefore adopted. For a fixed auxiliary variable $\beta^{(t)}$, by defining
\begin{equation}
\mathbf{a}_n^{(t)}\triangleq\Re\!\left\{(\beta^{(t)})^*\mathbf{h}_n\right\},
\end{equation}
and
\begin{equation}
b_n^{(t)}\triangleq |\beta^{(t)}|^2>0,
\end{equation}
the objective function in \eqr{dynproblem4} can be written as follows:
\begin{equation}
g(\boldsymbol{\alpha}_n,\beta^{(t)})=2(\mathbf{a}_n^{(t)})^\mathsf{T}\boldsymbol{\alpha}_n-b_n^{(t)}\|\boldsymbol{\alpha}_n\|_2^2.
\end{equation}
Therefore, problem \eqr{dynproblem4} is equivalent to the following form:
\begin{subequations}
\begin{empheq}{align}
\max_{\boldsymbol{\alpha}_n}\quad & 2(\mathbf{a}_n^{(t)})^\mathsf{T}\boldsymbol{\alpha}_n-b_n^{(t)}\|\boldsymbol{\alpha}_n\|_2^2\\
\textrm{s.t.}\quad & \mathbf{0}\preceq \boldsymbol{\alpha}_n \preceq \mathbf{1}.
\end{empheq}
\label{dynproblem5}
\end{subequations}\vspace{-2mm}\\
Since the objective function can be expressed as follows:
\begin{equation}
2(\mathbf{a}_n^{(t)})^\mathsf{T}\boldsymbol{\alpha}_n-b_n^{(t)}\|\boldsymbol{\alpha}_n\|_2^2=\sum_{k=1}^K \left(2a_{k,n}^{(t)}\alpha_{k,n}-b_n^{(t)}\alpha_{k,n}^2\right),
\end{equation}
problem \eqr{dynproblem5} is a separable concave quadratic program. As a result, problem \eqr{dynproblem5} can be decomposed into $K$ independent scalar subproblems. For any $k\in\mathcal{K}$, the optimal $\alpha_{k,n}$ can be obtained by solving the following problem:
\begin{subequations}
\begin{empheq}{align}
\max_{\alpha_{k,n}}\quad & 2a_{k,n}^{(t)}\alpha_{k,n}-b_n^{(t)}\alpha_{k,n}^2\\
\textrm{s.t.}\quad & 0\le\alpha_{k,n}\le 1.
\end{empheq}
\label{dynproblem6}
\end{subequations}\vspace{-2mm}\\
In problem \eqr{dynproblem6}, the objective is a concave quadratic function
\begin{equation}
\phi(\alpha_{k,n})=2a_{k,n}^{(t)}\alpha_{k,n}-b_n^{(t)}\alpha_{k,n}^2,
\end{equation}
where $b_n^{(t)}>0$. By ignoring constraint (\ref{dynproblem6}b), the stationary point can be obtained by setting the first-order derivative to zero, i.e.,
\begin{equation}
\frac{\partial\phi(\alpha_{k,n})}{\partial\alpha_{k,n}}=2a_{k,n}^{(t)}-2b_n^{(t)}\alpha_{k,n}=0,
\end{equation}
which leads to
\begin{equation}\label{solution}
\alpha_{k,n}=\frac{a_{k,n}^{(t)}}{b_n^{(t)}}.
\end{equation}
Note that the second-order derivative satisfies
\begin{equation}
\frac{\partial^2\phi(\alpha_{k,n})}{\partial\alpha_{k,n}^2}=-2b_n^{(t)}<0,
\end{equation}
which indicates that $\phi(\alpha_{k,n})$ is strictly concave. Therefore, \eqr{solution} gives the unique optimal solution of the unconstrained problem. By incorporating constraint (\ref{dynproblem6}b), the closed-form update can be presented as follows:
\begin{equation}\label{alpha2}
\alpha_{k,n}^{(t+1)}=\min\left\{1,\max\left\{0,\frac{a_{k,n}^{(t)}}{b_n^{(t)}}\right\}\right\}, \ \forall k\in\mathcal{K}.
\end{equation}
After obtaining $\boldsymbol{\alpha}_n^{(t+1)}$, $\beta^{(t+1)}$ is updated according to \eqref{beta2}. The QoS requirement can be verified by checking whether $g(\boldsymbol{\alpha}_n^{(t+1)},\beta^{(t+1)})\ge \gamma_\mathrm{min}$. If the condition is not satisfied, the target rate $R_\mathrm{min}$ is declared infeasible for user $n$ under the current system parameters.

The above closed-form update leads to a simple alternating optimization procedure. Given a nonzero initialization $\boldsymbol{\alpha}_n^{(0)}\neq\mathbf{0}$ and the iteration index $t\in\{0,1,\dots\}$, the aperture vector and the auxiliary variable are updated based on \eqr{alpha2} and \eqr{beta2}, respectively. This procedure is repeated until convergence.
\section{Simulation Results}
\begin{table}[t]
\centering
\caption{Simulation Parameters}\vspace{-1mm}
\label{parameter}
\begin{tabular}{lc}
\hline
\textbf{Parameter} & \textbf{Value} \\ \hline
Carrier frequency ($f_c$) & $3.5$~GHz\\
Noise power ($\sigma^2$) & $-64$~dBm\\ 
Attenuation constant ($\kappa$) & $0.1$~dB/m \\
Relative permittivity ($\varepsilon_r$) & $1.26$ \\
Cable height ($d$) & $3$~m \\
Region size ($D_x$) & $50$~m \\
Number of scatterers ($L$) & $20$ \\
Scattering path gain ($\delta_\ell$) & $\delta_\ell \sim \mathcal{CN}(0,1)$ \\ \hline
\end{tabular}
\end{table}

In this section, simulation results are provided to evaluate the performance of the proposed LCX pinching-antenna system. In the simulation, users are randomly distributed on the ground plane within the service region, while the scatterers are randomly placed on the surrounding walls to model reflected propagation paths. For sum rate calculation, only users whose achievable data rates exceed the target rate are included, and the outage probability is measured as the fraction of users whose achievable data rates fall below the target rate. As benchmarks, two conventional systems are considered: i) a fixed-antenna system with $K$ antennas deployed at the BS with an inter-element spacing of $\lambda/2$ at the same height of $3$~m, and ii) a conventional LCX system in which all slots are fully open. For the LCX based schemes, both binary slot activation and the proposed aperture adjustment adopt the conventional LCX configuration as the initial state. The main simulation parameters are summarized in Table~\ref{parameter}.

\begin{figure}[!t]
\centering{
\subfigure[Sum Rate]{\centering{\includegraphics[width=82mm]{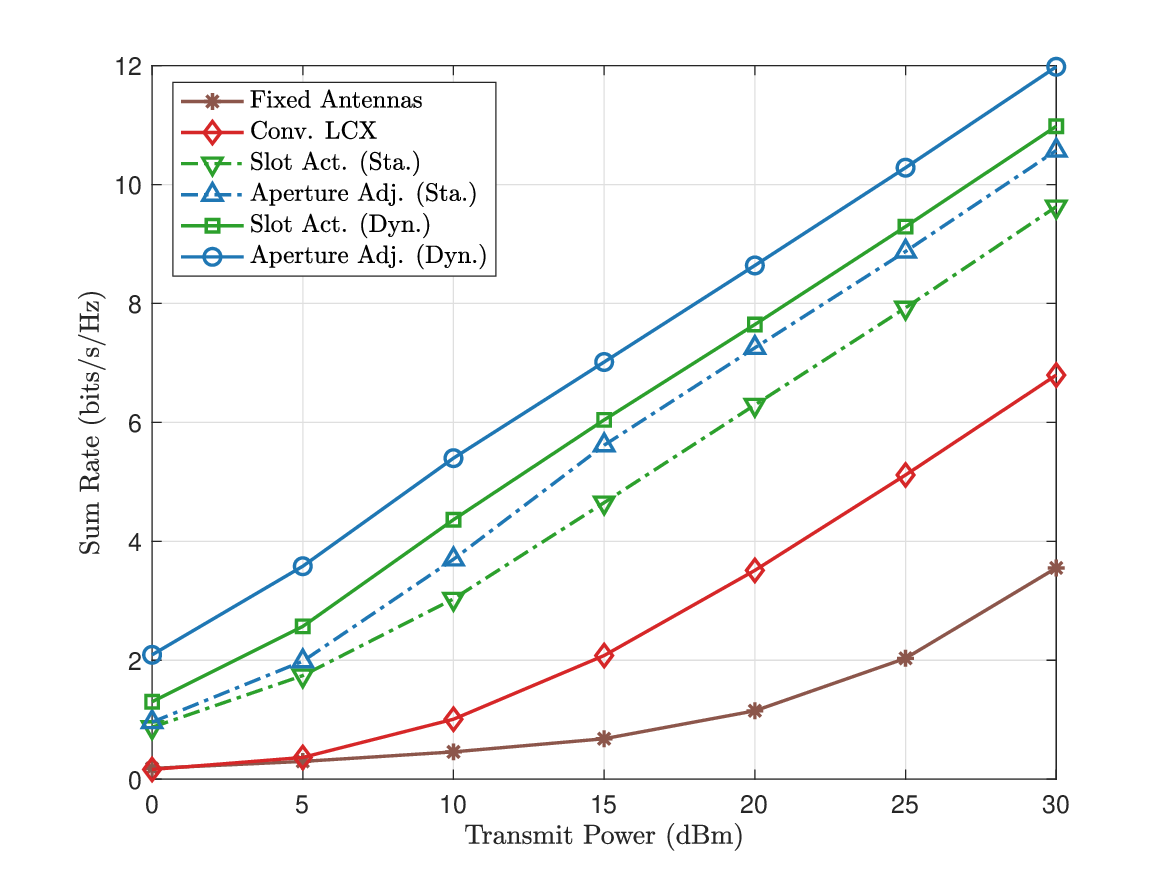}}}\vspace{-2mm}
\subfigure[Outage Probability]{\centering{\includegraphics[width=82mm]{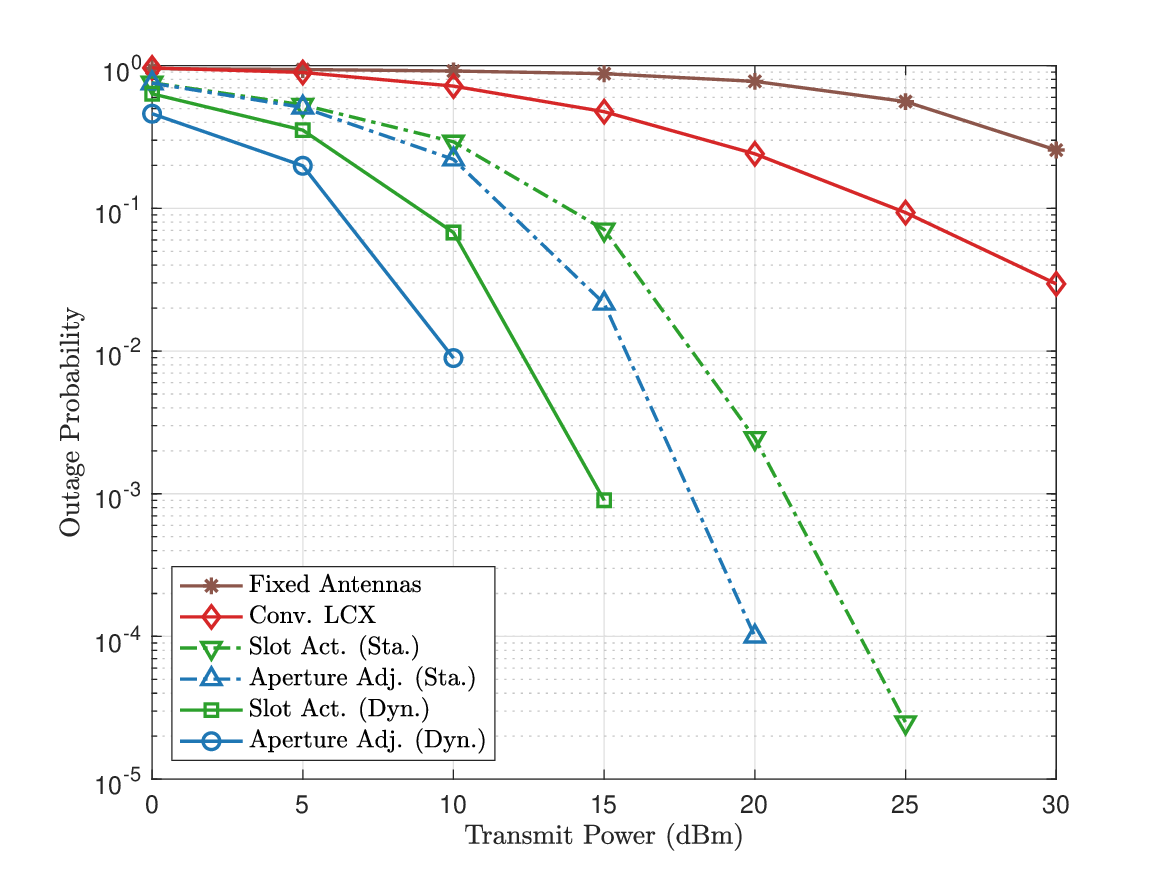}}}}\vspace{-2mm}
\caption{Impact of the transmit power on the sum rate and outage probability, where $D_y=20$~m, $K=50$, $N=4$, and $R_\mathrm{min}=0.5$~bits/s/Hz.}\vspace{-4mm}
\label{result1}
\end{figure}

\fref{result1} illustrates the impact of the transmit power on the system performance in terms of the sum rate and the outage probability. As the transmit power increases, the achievable sum rate improves monotonically, while the outage probability decreases correspondingly. Among the baselines, the conventional LCX system outperforms the fixed-antenna system in both metrics, owing to its distributed radiation characteristic that provides more uniform coverage across the service region. By enabling controllable radiation at specific slot locations, both slot activation and aperture adjustment further improve the performance relative to the conventional LCX baseline. In particular, slot-level radiation control enhances signal combining at the users, thereby increasing the sum rate and reducing the outage probability. Aperture adjustment consistently outperforms slot activation due to the additional flexibility offered by continuous aperture control, which enlarges the feasible design space and enables more refined amplitude re-weighting across slots, as supported by Propositions~1 and~2. In addition, the dynamic schemes achieve higher sum rates and lower outage probabilities than the corresponding static schemes, since user-specific design in orthogonal TDMA slots offers greater flexibility than a common aperture vector shared by all users. Moreover, the curves of slot activation and aperture adjustment enter the approximately linear growth region earlier than those of the conventional LCX and fixed-antenna baselines, indicating that slot-level control improves the effective channel combining gain and allows the system to reach the high-SNR regime at a lower transmit power.

\begin{figure}[!t]
\centering{
\subfigure[Sum Rate]{\centering{\includegraphics[width=82mm]{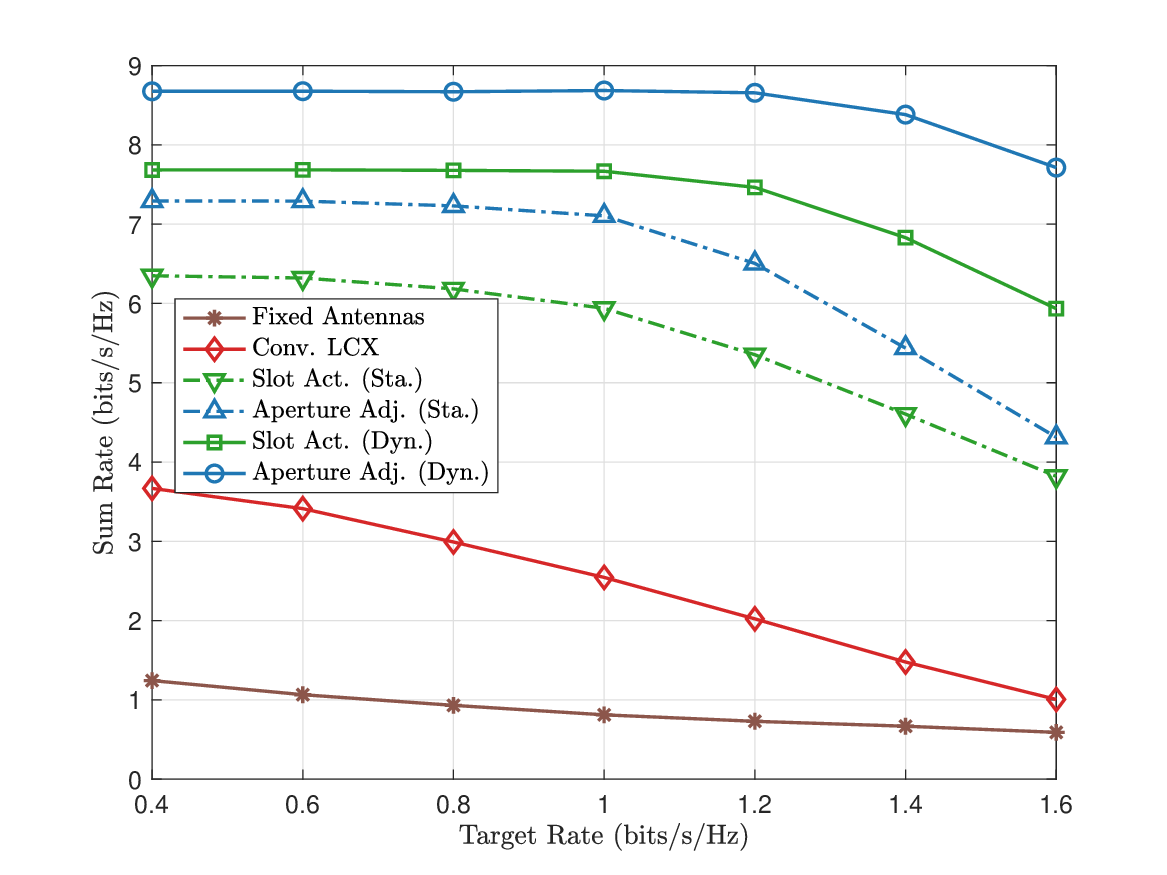}}}\vspace{-2mm}
\subfigure[Outage Probability]{\centering{\includegraphics[width=82mm]{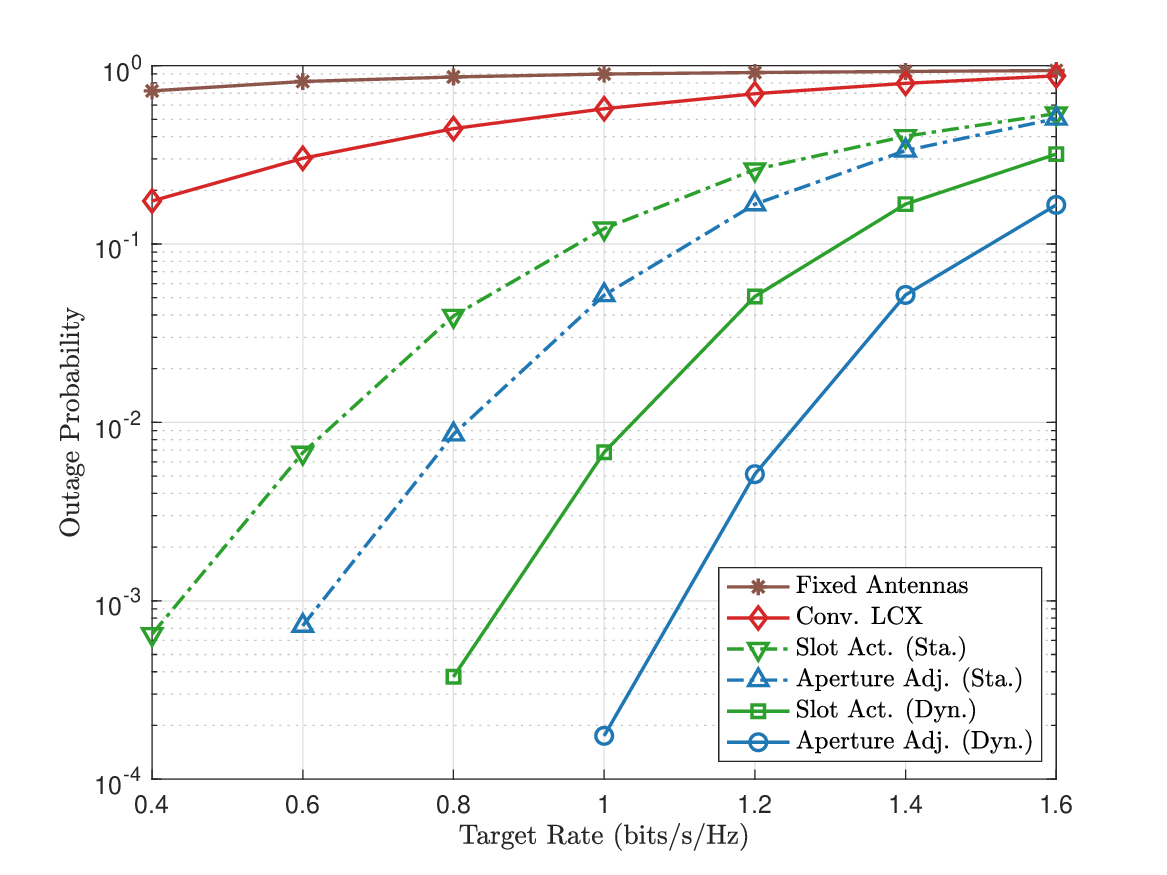}}}}\vspace{-2mm}
\caption{Impact of the target rate on the sum rate and outage probability, where $D_y=20$~m, $K=50$, $N=4$, and $P_t=20$~dBm.}\vspace{-4mm}
\label{result2}
\end{figure}

\fref{result2} presents the sum rate and outage probability as functions of the target rate $R_\mathrm{min}$. It can be observed that, for all schemes, increasing the target rate leads to a reduction in the sum rate and a corresponding increase in the outage probability. Compared with the conventional LCX and fixed-antenna baselines, both slot activation and aperture adjustment achieve significantly improved performance. In particular, when the target rate is relatively small, these two schemes are able to maintain the sum rate at nearly the same level, since most users can still satisfy the QoS requirement. As the target rate further increases, however, fewer users can meet the constraint, and the sum rate of all schemes begins to decrease. Moreover, aperture adjustment consistently outperforms the corresponding slot activation schemes in both the static and dynamic cases due to its greater flexibility in slot amplitude control. It can also be observed that the performance gap between the static and dynamic schemes becomes larger as the target rate increases. This is because a higher target rate imposes more stringent QoS constraints, under which the static scheme is increasingly limited by the need to satisfy multiple users with a common aperture configuration, while the dynamic scheme can still adapt the aperture pattern to each user's channel.

\begin{figure}[!t]
\centering{
\subfigure[Sum Rate]{\centering{\includegraphics[width=82mm]{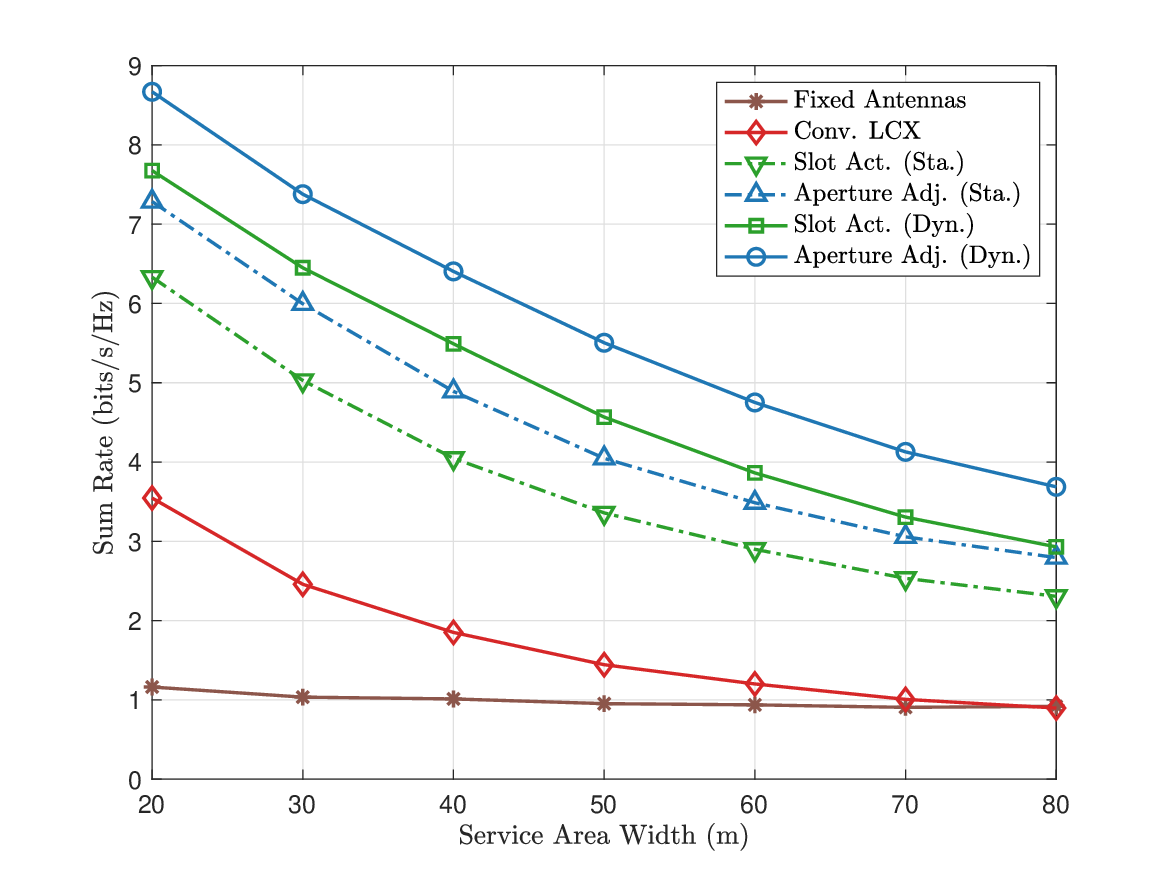}}}\vspace{-2mm}
\subfigure[Outage Probability]{\centering{\includegraphics[width=82mm]{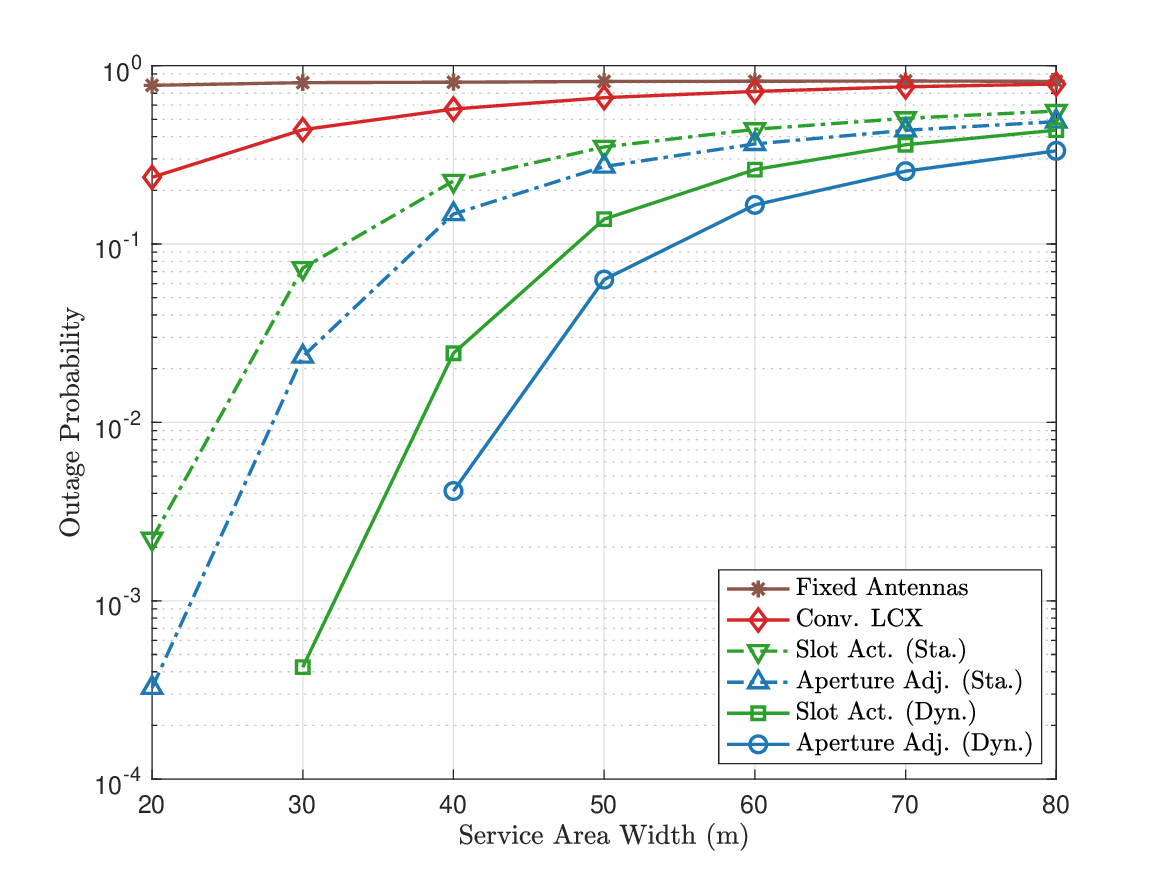}}}}\vspace{-2mm}
\caption{Impact of the width of the region on the sum rate and outage probability, where $K=50$, $N=4$, $P_t=20$~dBm, and $R_\mathrm{min}=0.5$~bits/s/Hz.}\vspace{-4mm}
\label{result3}
\end{figure}

\fref{result3} illustrates the impact of the service area width $D_y$ on the system performance. As $D_y$ increases, the sum rate decreases while the outage probability increases for all schemes. It can be observed that the performance of the fixed-antenna system is only mildly affected by this change, whereas the LCX based schemes experience a much more pronounced degradation. This is because, in the LCX pinching-antenna system, the radiated energy is mainly concentrated around the slots along the cable. Increasing $D_y$ generally implies that users are located farther away from the cable, resulting in weaker effective channel gains. Moreover, in the static case, the performance advantage of aperture adjustment over slot activation becomes less significant as $D_y$ increases, since the common aperture configuration must compromise among users with increasingly diverse channel conditions. In contrast, in the dynamic case, the performance gap between aperture adjustment and slot activation becomes larger as the service area expands, because dynamic aperture control enables user-specific radiation optimization and can better exploit the additional flexibility provided by continuous aperture adjustment.

\begin{figure}[!t]
\centering{
\subfigure[Sum Rate]{\centering{\includegraphics[width=82mm]{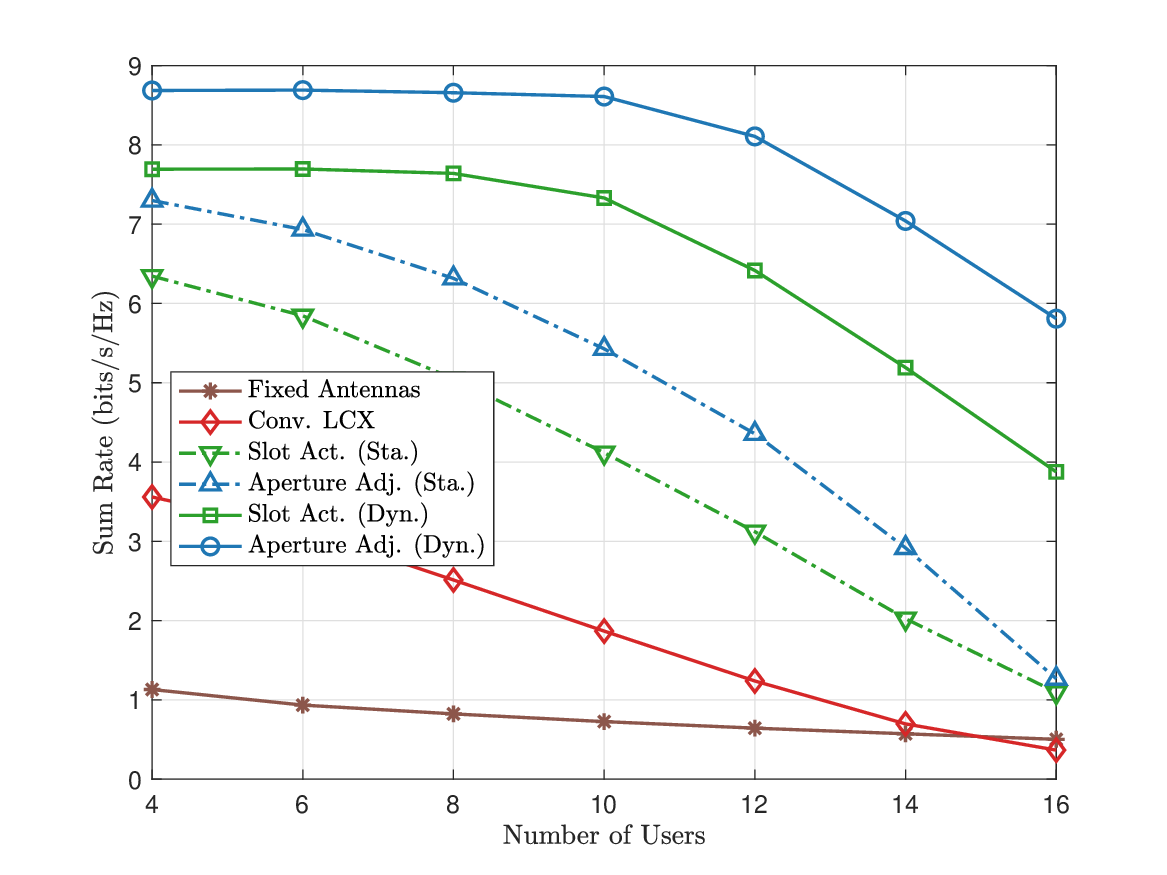}}}\vspace{-2mm}
\subfigure[Outage Probability]{\centering{\includegraphics[width=82mm]{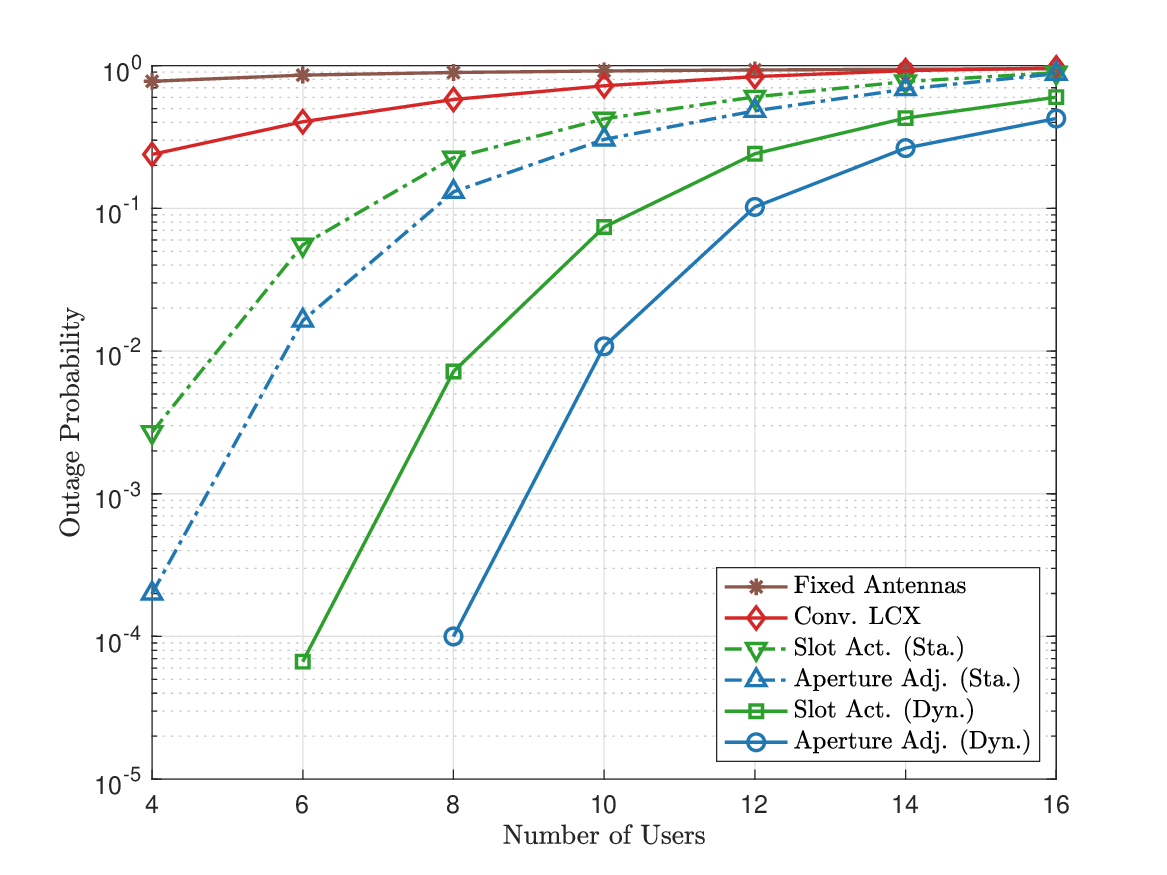}}}}\vspace{-2mm}
\caption{Impact of the number of users on the sum rate and outage probability, where $D_y=20$~m, $K=50$, $P_t=20$~dBm, and $R_\mathrm{min}=0.5$~bits/s/Hz.}\vspace{-4mm}
\label{result4}
\end{figure}

\fref{result4} shows the impact of the number of users on the sum rate and the outage probability. As the number of users increases, the sum rate decreases and the outage probability increases, since each user is allocated less time resource under the TDMA protocol and the QoS constraints become more difficult to satisfy simultaneously. It can be observed that the fixed-antenna baseline is mildly affected by the number of users. In contrast, the performance of the LCX based schemes degrades more noticeably, and the conventional LCX scheme even yields a lower sum rate than the fixed-antenna baseline when the number of users becomes large. For the static case, the performance of aperture adjustment gradually approaches that of slot activation as the user number increases, since the common aperture configuration leaves limited room for continuous aperture optimization. By contrast, in the dynamic case, the performance gap between aperture adjustment and slot activation becomes more pronounced with the user number, because user-specific aperture control makes the benefit of continuous slot-level amplitude adaptation more evident. 
\section{Conclusions}
This paper studied an LCX pinching-antenna system with adjustable slot apertures, where binary slot activation was extended to continuous aperture adjustment, inducing a continuous weighting effect of each slot on the corresponding channel amplitude. Analytical results characterized the performance gain of continuous aperture control and revealed the impact of channel coherence. Under the proposed framework, both static and dynamic TDMA schemes were developed, and the formulated sum rate maximization problems were efficiently solved. Simulation results demonstrated clear performance improvements over conventional slot activation, fixed-antenna, and traditional LCX systems, and validated the analytical insights. These findings highlight the potential of continuous aperture adjustment for enhancing distributed antenna systems in future wireless networks.
\section*{Appendix~A: Proof of Proposition~\ref{strict}}
This proposition can be proved by considering two cases for a fixed binary activation vector $\boldsymbol{\alpha}_\mathrm{bin}$ with active slot index set $\mathcal{S}$: i) decreasing the aperture adjustment coefficient of an active slot (partial aperture closing), and ii) increasing the aperture adjustment coefficient of an inactive slot (partial aperture opening). Since $\log_2(1+x)$ is strictly increasing, maximizing the achievable data rate of user $n$ is equivalent to maximizing the following normalized SNR metric:
\begin{equation}
\gamma_n(\boldsymbol{\alpha})=\frac{|\mathbf{h}_n^\mathsf{T}\boldsymbol{\alpha}|^2}{\|\boldsymbol{\alpha}\|_2^2}.
\end{equation}
For each case, a feasible continuous variation is constructed and shown to induce a strictly positive first-order increase in $\gamma_n(\boldsymbol{\alpha})$.
\subsubsection{Partial Closing of Active Slot Apertures}
With a given active slot $k\in\mathcal{S}$, let $\tau\in(0,1)$ denote a variation parameter. The resulting feasible aperture adjustment vector can be obtained as follows:
\begin{equation}
\boldsymbol{\alpha}(\tau)=\boldsymbol{\alpha}_\mathrm{bin}-\tau\mathbf{e}_k,
\end{equation}
where $\mathbf{e}_k$ is the $k$-th standard basis vector. This variation keeps all slot aperture adjustment coefficients within $[0,1]$ since $\alpha_k(\tau)=1-\tau\in(0,1)$ and all other components remain $0$ or $1$. By substituting $\boldsymbol{\alpha}(\tau)$ into $\gamma_n(\boldsymbol{\alpha})$, the resulting expression can be obtained. The numerator term is given by
\begin{equation}
\mathbf{h}_n^\mathsf{T}\boldsymbol{\alpha}(\tau)=\mathbf{h}_n^\mathsf{T}(\boldsymbol{\alpha}_\mathrm{bin}-\tau\mathbf{e}_k)=\mathbf{h}_n^\mathsf{T}\boldsymbol{\alpha}_\mathrm{bin}-\tau h_{k,n}.
\end{equation}
Moreover, the denominator term becomes
\begin{equation}
\|\boldsymbol{\alpha}(\tau)\|_2^2=\|\boldsymbol{\alpha}_\mathrm{bin}-\tau\mathbf{e}_k\|_2^2=\|\boldsymbol{\alpha}_\mathrm{bin}\|_2^2-2\tau\alpha_{\mathrm{bin},k}+\tau^2,
\end{equation}
where $\alpha_{\mathrm{bin},k}=1$ holds for $k\in\mathcal{S}$. As a result, defining $f(\tau)\triangleq \gamma_n(\boldsymbol{\alpha}(\tau))$ leads to the following equation:
\begin{equation}\label{f1}
f(\tau)=\frac{|\mathbf{h}_n^\mathsf{T}\boldsymbol{\alpha}_\mathrm{bin}-\tau h_{k,n}|^2}{\|\boldsymbol{\alpha}_\mathrm{bin}\|_2^2-2\tau+\tau^2}.
\end{equation}
By expanding the numerator, it follows that
\begin{align}\nonumber
&|\mathbf{h}_n^\mathsf{T}\boldsymbol{\alpha}_\mathrm{bin}-\tau h_{k,n}|^2\\
=&|\mathbf{h}_n^\mathsf{T}\boldsymbol{\alpha}_\mathrm{bin}|^2-2\tau\Re\!\left\{(\mathbf{h}_n^\mathsf{T}\boldsymbol{\alpha}_\mathrm{bin})h_{k,n}^*\right\}+\tau^2|h_{k,n}|^2.
\end{align}
Accordingly, \eqr{f1} can be rewritten as follows:
\begin{equation}
f(\tau)=\frac{|\mathbf{h}_n^\mathsf{T}\boldsymbol{\alpha}_\mathrm{bin}|^2\!\!-\!2\tau\Re\!\left\{\!(\mathbf{h}_n^\mathsf{T}\boldsymbol{\alpha}_\mathrm{bin})h_{k,n}^*\!\right\}\!+\!\tau^2|h_{k,n}|^2}{\|\boldsymbol{\alpha}_\mathrm{bin}\|_2^2-2\tau+\tau^2}.
\end{equation}
To characterize the first-order increase at $\tau=0$, define
\begin{equation}
u(\tau)\triangleq|\mathbf{h}_n^\mathsf{T}\boldsymbol{\alpha}_\mathrm{bin}|^2\!\!-\!2\tau\Re\!\left\{\!(\mathbf{h}_n^\mathsf{T}\boldsymbol{\alpha}_\mathrm{bin})h_{k,n}^*\!\right\}\!+\!\tau^2|h_{k,n}|^2,
\end{equation}
and
\begin{equation}
v(\tau)\triangleq\|\boldsymbol{\alpha}_\mathrm{bin}\|_2^2-2\tau+\tau^2,	
\end{equation}
such that $f(\tau)=u(\tau)/v(\tau)$. The values at $\tau=0$ are
\begin{equation}
\left\{\begin{aligned}
u(0)&=|\mathbf{h}_n^\mathsf{T}\boldsymbol{\alpha}_\mathrm{bin}|^2,\\
u'(0)&=-2\Re\!\left\{(\mathbf{h}_n^\mathsf{T}\boldsymbol{\alpha}_\mathrm{bin})h_{k,n}^*\right\},\\
v(0)&=\|\boldsymbol{\alpha}_\mathrm{bin}\|_2^2,\\
v'(0)&=-2.
\end{aligned}\right.
\end{equation}
By applying the quotient rule, it gives
\begin{align}\nonumber
f'(0)&=\frac{u'(0)v(0)-u(0)v'(0)}{v(0)^2}\\
&=\frac{2\!\left(|\mathbf{h}_n^\mathsf{T}\boldsymbol{\alpha}_\mathrm{bin}|^2\!\!-\!\|\boldsymbol{\alpha}_\mathrm{bin}\|_2^2\Re\!\left\{(\mathbf{h}_n^\mathsf{T}\boldsymbol{\alpha}_\mathrm{bin})h_{k,n}^*\right\}\right)}{\|\boldsymbol{\alpha}_\mathrm{bin}\|_2^4}.
\end{align}
Since $\|\boldsymbol{\alpha}_\mathrm{bin}\|_2^2=|\mathcal{S}|>0$, $f'(0)>0$ holds if
\begin{equation}\label{fcond1}
\Re\!\left\{(\mathbf{h}_n^\mathsf{T}\boldsymbol{\alpha}_\mathrm{bin})h_{k,n}^*\right\}<\frac{|\mathbf{h}_n^\mathsf{T}\boldsymbol{\alpha}_\mathrm{bin}|^2}{\|\boldsymbol{\alpha}_\mathrm{bin}\|_2^2}.
\end{equation}
Based on the continuity of $f(\tau)$, if inequality \eqr{fcond1} holds, there exists a sufficiently small $\tau>0$ such that $f(\tau)>f(0)$, i.e.,
\begin{equation}
\gamma_n(\boldsymbol{\alpha}(\tau))>\gamma_n(\boldsymbol{\alpha}_\mathrm{bin}).
\end{equation}
Since $R_n(\boldsymbol{\alpha})=\log_2(1+\frac{P_t}{\sigma^2}\gamma_n(\boldsymbol{\alpha}))$ is strictly increasing in $\gamma_n(\boldsymbol{\alpha})$, it follows that
\begin{equation}
R_n(\boldsymbol{\alpha}(\tau))>R_n(\boldsymbol{\alpha}_\mathrm{bin}).
\end{equation}
This establishes strict improvement under condition (\ref{strictcond}a).
\subsubsection{Partial Opening of Inactive Slot Apertures}
Consider an inactive slot $k\notin\mathcal{S}$. For a small variation parameter $\tau\in(0,1)$, the $k$-th aperture adjustment coefficient is increased from zero by defining
\begin{equation}
\boldsymbol{\alpha}(\tau)=\boldsymbol{\alpha}_\mathrm{bin}+\tau\mathbf{e}_k,
\end{equation}
which corresponds to partially opening the aperture of slot $k$. Since $\alpha_{\mathrm{bin},k}=0$ for $k\notin\mathcal{S}$, one has $\alpha_k(\tau)=\tau\in(0,1)$, while all other components remain unchanged in $\{0,1\}$, and hence $\boldsymbol{\alpha}(\tau)\in[0,1]^K$ is feasible. Substituting $\boldsymbol{\alpha}(\tau)$ into $\gamma_n(\boldsymbol{\alpha})$ gives
\begin{equation}
\mathbf{h}_n^\mathsf{T}\boldsymbol{\alpha}(\tau)=\mathbf{h}_n^\mathsf{T}(\boldsymbol{\alpha}_\mathrm{bin}+\tau\mathbf{e}_k)=\mathbf{h}_n^\mathsf{T}\boldsymbol{\alpha}_\mathrm{bin}+\tau h_{k,n},
\end{equation}
and
\begin{equation}
\|\boldsymbol{\alpha}(\tau)\|_2^2=\|\boldsymbol{\alpha}_\mathrm{bin}\|_2^2+2\tau\alpha_{\mathrm{bin},k}+\tau^2=\|\boldsymbol{\alpha}_\mathrm{bin}\|_2^2+\tau^2,
\end{equation}
where the latter equality follows from $\alpha_{\mathrm{bin},k}=0$. As a result, defining $f(\tau)\triangleq \gamma_n(\boldsymbol{\alpha}(\tau))$ gives
\begin{equation}\label{f2}
f(\tau)=\frac{|\mathbf{h}_n^\mathsf{T}\boldsymbol{\alpha}_\mathrm{bin}+\tau h_{k,n}|^2}{\|\boldsymbol{\alpha}_\mathrm{bin}\|_2^2+\tau^2}.
\end{equation}
The numerator can be rewritten as follows:
\begin{align}\nonumber
&|\mathbf{h}_n^\mathsf{T}\boldsymbol{\alpha}_\mathrm{bin}+\tau h_{k,n}|^2\\
=&|\mathbf{h}_n^\mathsf{T}\boldsymbol{\alpha}_\mathrm{bin}|^2\!+2\tau\Re\!\left\{(\mathbf{h}_n^\mathsf{T}\boldsymbol{\alpha}_\mathrm{bin})h_{k,n}^*\right\}+\tau^2|h_{k,n}|^2.
\end{align}
Hence, \eqr{f2} becomes
\begin{equation}
f(\tau)=\frac{|\mathbf{h}_n^\mathsf{T}\boldsymbol{\alpha}_\mathrm{bin}|^2\!\!+\!2\tau\Re\!\left\{\!(\mathbf{h}_n^\mathsf{T}\boldsymbol{\alpha}_\mathrm{bin})h_{k,n}^*\!\right\}\!+\!\tau^2|h_{k,n}|^2}{\|\boldsymbol{\alpha}_\mathrm{bin}\|_2^2+\tau^2}.
\end{equation}
To examine the first-order behavior around  $\tau=0$, define
\begin{equation}
u(\tau)\triangleq |\mathbf{h}_n^\mathsf{T}\boldsymbol{\alpha}_\mathrm{bin}|^2\!\!+\!2\tau\Re\!\left\{(\mathbf{h}_n^\mathsf{T}\boldsymbol{\alpha}_\mathrm{bin})h_{k,n}^*\right\}\!+\!\tau^2|h_{k,n}|^2,
\end{equation}
and
\begin{equation}
v(\tau)\triangleq\|\boldsymbol{\alpha}_\mathrm{bin}\|_2^2+\tau^2.
\end{equation}
It follows that $f(\tau)=u(\tau)/v(\tau)$. At $\tau=0$, the following equations hold
\begin{equation}
\left\{\begin{aligned}
u(0)&=|\mathbf{h}_n^\mathsf{T}\boldsymbol{\alpha}_\mathrm{bin}|^2,\\
u'(0)&=2\Re\!\left\{(\mathbf{h}_n^\mathsf{T}\boldsymbol{\alpha}_\mathrm{bin})h_{k,n}^*\right\},\\
v(0)&=\|\boldsymbol{\alpha}_\mathrm{bin}\|_2^2,\\
v'(0)&=0.
\end{aligned}\right.
\end{equation}
Applying the quotient rule gives
\begin{equation}
f'(0)=\frac{u'(0)v(0)-u(0)v'(0)}{v(0)^2}=\frac{2\Re\!\left\{(\mathbf{h}_n^\mathsf{T}\boldsymbol{\alpha}_\mathrm{bin})h_{k,n}^*\right\}}{\|\boldsymbol{\alpha}_\mathrm{bin}\|_2^2}.
\end{equation}
Therefore, $f'(0)>0$ holds if
\begin{equation}
\Re\!\left\{(\mathbf{h}_n^\mathsf{T}\boldsymbol{\alpha}_\mathrm{bin})h_{k,n}^*\right\}>0.
\end{equation}
By the continuity of $f(\tau)$, for sufficiently small $\tau>0$, it holds that $f(\tau)>f(0)$, i.e.,
\begin{equation}
\gamma_n(\boldsymbol{\alpha}(\tau))>\gamma_n(\boldsymbol{\alpha}_\mathrm{bin}),
\end{equation}
which indicates the strict improvement of the data rate under condition (\ref{strictcond}b), as follows:
\begin{equation}
R_n(\boldsymbol{\alpha}(\tau))>R_n(\boldsymbol{\alpha}_\mathrm{bin}).
\end{equation}

By combining these two conditions, if either (\ref{strictcond}a) or (\ref{strictcond}b) holds,  there exists a feasible continuous aperture adjustment vector that strictly improves the data rate over the binary activation vector. This proof is completed.\QEDA
\section*{Appendix~B: Proof of Proposition~\ref{upperbound}}
In the conventional LCX system, all slots radiate continuously. This scenario corresponds to a particular binary slot activation pattern in which all slots are active, i.e., $\alpha_{\mathrm{bin},k}=1$ for all $k\in\mathcal{K}$. In this case, let $\boldsymbol{\alpha}_{\mathrm{bin}}=\mathbf{1}\in\{0,1\}^K$ denote the activation vector. It follows that $\|\boldsymbol{\alpha}_{\mathrm{bin}}\|_2^2=K$ and $\mathbf{h}_n^\mathsf{T}\boldsymbol{\alpha}_{\mathrm{bin}}=\sum_{k\in\mathcal{K}}h_{k,n}$. Hence, the normalized SNR metric under the conventional LCX baseline is given by
\begin{equation}
\gamma_n(\boldsymbol{\alpha}_\mathrm{bin})=\frac{1}{K}\left|\sum_{k\in\mathcal{K}}h_{k,n}\right|^2.
\end{equation}
By performing aperture adjustment, the coefficient vector $\boldsymbol{\alpha}$ is designed over all slots, subject to $\alpha_k\in[0,1]$ for all $k\in\mathcal{K}$ and $\boldsymbol{\alpha}\neq\mathbf{0}$. For any feasible vector $\boldsymbol{\alpha}$ satisfying these constraints, the normalized SNR is given by
\begin{equation}
\gamma_n(\boldsymbol{\alpha})=\frac{\left|\sum_{k\in\mathcal{K}}\alpha_k h_{k,n}\right|^2}{\sum_{k\in\mathcal{K}}\alpha_k^2}=\frac{|\boldsymbol{\alpha}^\mathsf{T}\mathbf{h}_n|^2}{\|\boldsymbol{\alpha}\|_2^2}.
\end{equation}
Based on the Cauchy-Schwarz inequality, the following inequality can be obtained:
\begin{align}\nonumber
&|\boldsymbol{\alpha}^\mathsf{T}\mathbf{h}_n|^2\le\|\boldsymbol{\alpha}\|_2^2\|\mathbf{h}_n\|_2^2\\
\Rightarrow &\left|\sum_{k\in\mathcal{K}}\alpha_k h_{k,n}\right|^2\le\left(\sum_{k\in\mathcal{K}}\alpha_k^2\right)\left(\sum_{k\in\mathcal{K}}|h_{k,n}|^2\right).
\end{align}
Since $\boldsymbol{\alpha}\neq\mathbf{0}$, it follows that $\sum_{k\in\mathcal{K}}\alpha_k^2>0$. Dividing both sides by $\sum_{k\in\mathcal{K}}\alpha_k^2$ gives the following inequality:
\begin{equation}
\gamma_n(\boldsymbol{\alpha})=\frac{\left|\sum_{k\in\mathcal{K}}\alpha_k h_{k,n}\right|^2}{\sum_{k\in\mathcal{K}}\alpha_k^2}\le\sum_{k\in\mathcal{K}}|h_{k,n}|^2,
\end{equation}
and therefore
\begin{equation}\label{gamma1}
\gamma_n(\boldsymbol{\alpha})\le\sum_{k\in\mathcal{K}}|h_{k,n}|^2.
\end{equation}
The corresponding rate gain over the conventional LCX baseline is given by
\begin{equation}\label{deltar}
\Delta R_n\triangleq\log_2\!\left(1+\frac{P_t}{\sigma^2}\gamma_n(\boldsymbol{\alpha})\right)-\log_2\!\left(1+\frac{P_t}{\sigma^2}\gamma_n(\boldsymbol{\alpha}_\mathrm{bin})\right).
\end{equation}
Substituting \eqr{gamma1} into \eqr{deltar} gives
\begin{equation}\label{deltar1}
\Delta R_n\le\log_2\!\left(\frac{1+\frac{P_t}{\sigma^2}\sum_{k\in\mathcal{K}}|h_{k,n}|^2}{1+\frac{P_t}{\sigma^2}\frac{\left|\sum_{k\in\mathcal{K}}h_{k,n}\right|^2}{K}}\right).
\end{equation}

Under high SNR, $\log_2(1+x)\approx\log_2(x)$, and \eqr{deltar1} can be approximated as follows:
\begin{equation}\label{deltar2}
\Delta R_n\!\lesssim\!\log_2\!\!\left(\!\frac{K\!\sum_{k\in\mathcal{K}}\!|h_{k,n}|^2}{\left|\sum_{k\in\mathcal{K}}h_{k,n}\right|^2}\!\right)\!=\!-\!\log_2\!\!\left(\!\frac{\left|\sum_{k\in\mathcal{K}}h_{k,n}\right|^2}{K\!\sum_{k\in\mathcal{K}}\!|h_{k,n}|^2}\!\right).
\end{equation}
By applying the Cauchy-Schwarz inequality to $\mathbf{1}\in\mathbb{C}^{|\mathcal{K}|}$ and $\mathbf{h}_n$,  the following inequality can be obtained:
\begin{equation}\label{cs2}
|\mathbf{1}^\mathsf{H}\mathbf{h}_n|^2\le\|\mathbf{1}\|_2^2\|\mathbf{h}_n\|_2^2\Rightarrow\!\left|\sum_{k\in\mathcal{K}}\!h_{k,n}\right|^2\!\!\le K\!\sum_{k\in\mathcal{K}}\!|h_{k,n}|^2,
\end{equation}
since $\|\mathbf{1}\|^2=K$ and $\|\mathbf{h}_n\|^2=\sum_{k\in\mathcal{K}}|h_{k,n}|^2$. Provided that $\sum_{k\in\mathcal{K}}|h_{k,n}|^2>0$, inequality \eqr{cs2} indicates that the fraction inside the logarithm in \eqr{deltar2} is at most one, and thus the derived upper bound is non-negative, i.e., 
\begin{equation}
0\le\frac{\left|\sum_{k\in\mathcal{K}}h_{k,n}\right|^2}{K\!\sum_{k\in\mathcal{K}}|h_{k,n}|^2}\le 1.
\end{equation}

Since the binary activation vector $\boldsymbol{\alpha}_{\mathrm{bin}}=\mathbf{1}$ is feasible under aperture adjustment, i.e., $\boldsymbol{\alpha}_{\mathrm{bin}}\in[0,1]^K$, there exists a feasible aperture adjustment vector $\boldsymbol{\alpha}$ such that
\begin{equation}
\gamma_n(\boldsymbol{\alpha})\ge\gamma_n(\boldsymbol{\alpha}_\mathrm{bin}),
\end{equation}
and therefore,
\begin{equation}
\Delta R_n \ge 0.
\end{equation}
Under high SNR, the rate gain satisfies
\begin{equation}
0\le\Delta R_n\lesssim-\log_2\left(\frac{\left|\sum_{k\in\mathcal{K}}h_{k,n}\right|^2}{K\!\sum_{k\in\mathcal{K}}|h_{k,n}|^2}\right).
\end{equation}
This proposition is proven.\QEDA
\bibliographystyle{IEEEtran}
\bibliography{KaidisBib}

@STRING{IEEE_J_WCOML       	= "{IEEE} Wireless Commun. Lett."}

@STRING{IEEE_J_JSAC       	= "{IEEE} J. Sel. Areas Commun."}

@STRING{IEEE_J_COM        	= "{IEEE} Trans. Commun."}

@STRING{IEEE_J_SP         	= "{IEEE} Trans. Signal Process."}

@STRING{IEEE_J_WCOM       	= "{IEEE} Trans. Wireless Commun."}

@STRING{IEEE_J_VT         	= "{IEEE} Trans. Veh. Technol."}

@STRING{IEEE_M_COM        	= "{IEEE} Commun. Mag."}

@STRING{IEEE_WM_COM        	= "{IEEE} Wireless Commun."}

@STRING{IEEE_OJ_COMS       	= "{IEEE} Open J. Commun. Soc."}

@STRING{IEEE_J_IOT 			= "{IEEE} Internet Things J."}

@article{kaidi2026generalized,
  title={Leaky Coaxial Cable based Generalized Pinching-Antenna Systems with Dual-Port Feeding},
  author={Wang, Kaidi and Ding, Zhiguo and So, Daniel KC},
  journal={arXiv preprint arXiv:2602.21856},
  year={2026}
}

@ARTICLE{chen2025pin,
  author={Chen, Jung-Chieh and Wu, Po-Ching and Wong, Kai-Kit},
  journal=IEEE_OJ_COMS, 
  title={Dynamic Placement of Pinching Antennas for Multicast {MU}-{MISO} Downlinks}, 
  year={2025},
  volume={6},
  number={},
  pages={5611-5625},
  doi={10.1109/OJCOMS.2025.3582895}}

@ARTICLE{tyrovolas2026pin,
  author={Tyrovolas, Dimitrios and Tegos, Sotiris A. and Xiao, Yue and Diamantoulakis, Panagiotis D. and Ioannidis, Sotiris and Liaskos, Christos K. and Karagiannidis, George K. and Asimonis, Stylianos D.},
  journal=IEEE_J_IOT, 
  title={How many Pinching Antennas are Enough?}, 
  year={2026},
  volume={},
  number={},
  pages={1-1},
  doi={10.1109/JIOT.2026.3668307}}

@ARTICLE{xu2026pass,
  author={Xu, Xiaoxia and Mu, Xidong and Wang, Zhaolin and Liu, Yuanwei and Nallanathan, Arumugam},
  journal={IEEE Transactions on Communications}, 
  title={Pinching-Antenna Systems ({PASS}): Power Radiation Model and Optimal Beamforming Design}, 
  year={2026},
  volume={74},
  number={},
  pages={2160-2175},
  doi={10.1109/TCOMM.2025.3636083}}

@ARTICLE{zhou2026pin,
  author={Zhou, Enzhi and Cui, Jingjing and Liu, Ziyue and Ding, Zhiguo and Fan, Pingzhi},
  journal=IEEE_J_WCOM, 
  title={Joint Transmission for Cellular Networks With Pinching Antennas: System Design and Analysis}, 
  year={2026},
  volume={25},
  number={},
  pages={10175-10190},
  doi={10.1109/TWC.2026.3652240}}

@article{xie2026pinching,
  title={Pinching Antennas in Blockage-Aware Environments: Modeling, Design, and Optimization},
  author={Xie, Ximing and Fang, Fang and Ding, Zhiguo and Wang, Xianbin},
  journal={arXiv preprint arXiv:2601.01277},
  year={2026}
}

@article{ding2025analytical,
  title={Analytical optimization for antenna placement in pinching-antenna systems},
  author={Ding, Zhiguo and Poor, H Vincent},
  journal={arXiv preprint arXiv:2507.13307},
  year={2025}
}

@ARTICLE{shen2018fp,
  author={Shen, Kaiming and Yu, Wei},
  journal=IEEE_J_SP, 
  title={Fractional Programming for Communication Systems—Part I: Power Control and Beamforming}, 
  year={2018},
  volume={66},
  number={10},
  pages={2616-2630},
  doi={10.1109/TSP.2018.2812733}}

@ARTICLE{wang2001lcx1,
  author={Jun Hong Wang},
  journal={IEEE Microwave and Wireless Components Letters}, 
  title={Leaky coaxial cable with adjustable coupling loss for mobile communications in complex environments}, 
  year={2001},
  month={Aug.},
  volume={11},
  number={8},
  pages={346-348},
  doi={10.1109/7260.941785}}

@ARTICLE{siddiqui2020lcx,
  author={Siddiqui, Zeeshan and Sonkki, Marko and Tuhkala, Marko and Myllymäki, Sami},
  journal={IEEE Trans. Antennas Propag.}, 
  title={Periodically Slotted Coupled Mode Leaky Coaxial Cable With Enhanced Radiation Performance}, 
  year={2020},
  month={Nov.},
  volume={68},
  number={11},
  pages={7595-7600},
  doi={10.1109/TAP.2020.2990478}}

@INPROCEEDINGS{nagayama2022lcx1,
  author={Nagayama, Kenta and Zhu, Junjie and Hou, Pengcheng and Hou, Yafei and Denno, Satoshi},
  booktitle={2022 IEEE 4th Global Conference on Life Sciences and Technologies (LifeTech)}, 
  title={A Proposal of Spatial Modulation Using On/Off the Slots of Leaky Coaxial Cable}, 
  year={2022},
  volume={},
  number={},
  pages={289-290},
  doi={10.1109/LifeTech53646.2022.9754939}}

@misc{asplund2016leaky,
  title={Leaky coaxial cable having radiation slots that can be activated or deactivated},
  author={Asplund, Henrik and Berg, Jan-Erik and Medbo, Jonas},
  year={2016},
  month=aug # "~30",
  publisher={Google Patents},
  note={{US} Patent 9431716}}

@article{kaidi2025generalized,
  title={Generalized Pinching-Antenna Systems: A Leaky-Coaxial-Cable Perspective},
  author={Wang, Kaidi and Ding, Zhiguo and Hanzo, Lajos},
  journal={arXiv preprint arXiv:2512.04979},
  year={2025}
}

@article{ding2025edma,
  title={Environment division multiple access ({EDMA}): A feasibility study via pinching antennas},
  author={Ding, Zhiguo and Schober, Robert and Poor, H Vincent},
  journal={arXiv preprint arXiv:2511.03820},
  year={2025}
}

@ARTICLE{wang2025pa,
  author={Wang, Zhaolin and Ouyang, Chongjun and Mu, Xidong and Liu, Yuanwei and Ding, Zhiguo},
  journal=IEEE_J_COM, 
  title={Modeling and Beamforming Optimization for Pinching-Antenna Systems}, 
  year={2025},
  volume={73},
  number={12},
  pages={13904-13919},
  month={Oct.},
  doi={10.1109/TCOMM.2025.3621049}}

@INPROCEEDINGS{torrance1996lcx,
  author={Torrance, J.M. and Keller, T. and Hanzo, L.},
  booktitle={Proceedings of Vehicular Technology Conference - VTC}, 
  title={Multi-level modulation in the indoors leaky feeder environment}, 
  year={1996},
  volume={3},
  number={},
  pages={1554-1558 vol.3},
  doi={10.1109/VETEC.1996.504019}}

@article{xu2025generalized,
  title={Generalized Pinching-Antenna Systems: A Tutorial on Principles, Design Strategies, and Future Directions},
  author={Xu, Yanqing and Cui, Jingjing and Zhu, Yongxu and Ding, Zhiguo and Chang, Tsung-Hui and Schober, Robert and Wong, Vincent WS and Dobre, Octavia A and Karagiannidis, George K and Poor, H Vincent and others},
  journal={arXiv preprint arXiv:2510.14166},
  year={2025}}

@ARTICLE{wang2001lcx,
  author={Jun Hong Wang and Mei, K.K.},
  journal={IEEE Trans. Antennas Propag.}, 
  title={Theory and analysis of leaky coaxial cables with periodic slots}, 
  year={2001},
  month={Dec.},
  volume={49},
  number={12},
  pages={1723-1732},
  doi={10.1109/8.982452}}

@ARTICLE{yin2024lcx,
  author={Yin, Lu and Song, Tianzhu and Ni, Qiang and Xiao, Quanbin and Sun, Yuan and Guo, Wenfang},
  journal=IEEE_J_JSAC, 
  title={New Signal and Algorithms for {5G/6G} High Precision Train Positioning in Tunnel With Leaky Coaxial Cable}, 
  year={2024},
  volume={42},
  number={1},
  pages={223-238},
  doi={10.1109/JSAC.2023.3322790}}

@ARTICLE{morgan1999lcx,
  author={Morgan, S.P.},
  journal=IEEE_J_VT, 
  title={Prediction of indoor wireless coverage by leaky coaxial cable using ray tracing}, 
  year={1999},
  volume={48},
  number={6},
  pages={2005-2014},
  doi={10.1109/25.806793}}

@ARTICLE{tegos2025pin,
  author={Tegos, Sotiris A. and Diamantoulakis, Panagiotis D. and Ding, Zhiguo and Karagiannidis, George K.},
  journal=IEEE_J_WCOML, 
  title={Minimum Data Rate Maximization for Uplink Pinching-Antenna Systems}, 
  year={2025},
  month={Mar.},
  volume={14},
  number={5},
  pages={1516-1520},
  doi={10.1109/LWC.2025.3547956}}

@ARTICLE{xu2025pin2,
  author={Xu, Yanqing and Ding, Zhiguo and Cai, Donghong and Wong, Vincent W.S.},
  journal=IEEE_J_COM, 
  title={{QoS}-Aware {NOMA} Design for Downlink Pinching-Antenna Systems}, 
  year={2025},
  month={Sept.},
  volume={},
  number={},
  pages={1-1},
  doi={10.1109/TCOMM.2025.3605466}}

@ARTICLE{kaidi2025pin2,
  author={Wang, Kaidi and Ding, Zhiguo and Karagiannidis, George K.},
  journal=IEEE_J_WCOM, 
  title={Antenna Activation and Resource Allocation in Multi-Waveguide Pinching-Antenna Systems}, 
  year={2025},
  month={Sept.},
  volume={},
  number={},
  pages={1-1},
  doi={10.1109/TWC.2025.3608068}}

@ARTICLE{kaidi2025pin4,
  author={Wang, Kaidi and Ouyang, Chongjun and Liu, Yuanwei and Ding, Zhiguo},
  journal=IEEE_J_WCOML, 
  title={Pinching-Antenna Systems With {LoS} Blockages}, 
  year={2025},
  month={Sept.},
  volume={},
  number={},
  pages={1-1},
  doi={10.1109/LWC.2025.3614451}}

@article{liu2025pinching,
  author={Liu, Yuanwei and Wang, Zhaolin and Mu, Xidong and Ouyang, Chongjun and Xu, Xiaoxia and Ding, Zhiguo},
  journal=IEEE_M_COM, 
  title={Pinching-Antenna Systems: Architecture Designs, Opportunities, and Outlook}, 
  year={2025},
  month={Sept.},
  volume={},
  number={},
  pages={1-7},
  doi={10.1109/MCOM.001.2500037}}

@article{yang2025pinching,
  title={Pinching antennas: Principles, applications and challenges},
  author={Yang, Zheng and Wang, Ning and Sun, Yanshi and Ding, Zhiguo and Schober, Robert and Karagiannidis, George K and Wong, Vincent WS and Dobre, Octavia A},
  journal=IEEE_WM_COM,
  year={2025},
  month={Oct.},
  volume={},
  number={},
  pages={1-10},
}

@ARTICLE{xu2025pin,
  author={Xu, Yanqing and Ding, Zhiguo and Karagiannidis, George K.},
  journal=IEEE_J_WCOML, 
  title={Rate Maximization for Downlink Pinching-Antenna Systems}, 
  year={2025},
  volume={},
  number={},
  pages={1-1},
  doi={10.1109/LWC.2025.3543889}}

@ARTICLE{kaidi2025pin,
  author={Wang, Kaidi and Ding, Zhiguo and Schober, Robert},
  journal=IEEE_J_WCOML, 
  title={Antenna Activation for {NOMA} Assisted Pinching-Antenna Systems}, 
  year={2025},
  month={Mar.},
  volume={14},
  number={5},
  pages={1526-1530},
  doi={10.1109/LWC.2025.3548280}}

@article{wu2024fluid,
  author={Wu, Tuo and Zhi, Kangda and Yao, Junteng and Lai, Xiazhi and Zheng, Jianchao and Niu, Hong and Elkashlan, Maged and Wong, Kai-Kit and Chae, Chan-Byoung and Ding, Zhiguo and Karagiannidis, George K. and Debbah, Mérouane and Yuen, Chau},
  journal=IEEE_WM_COM, 
  title={Fluid Antenna Systems Enabling {6G}: Principles, Applications, and Research Directions}, 
  year={2025},
  volume={},
  number={},
  pages={1-9},
 
  doi={10.1109/MWC.2025.3629597}}

@article{suzuki2022pinching,
  title={Pinching Antenna-Using a Dielectric Waveguide as an Antenna},
  author={Fukuda, Atsushi and Yamamoto, Hiroto and  Okazaki, Hiroshi and Suzuki, Yasunori and Kawai, Kunihiro},
  journal={NTT DOCOMO Technical J.},
  volume={23},
  number={3},
  pages={5--12},
  year={2022},
  month = {Jan.}}

@ARTICLE{ma2024movable,
  author={Ma, Wenyan and Zhu, Lipeng and Zhang, Rui},
  journal=IEEE_J_WCOM, 
  title={{MIMO} Capacity Characterization for Movable Antenna Systems}, 
  year={2024},
  volume={23},
  number={4},
  pages={3392-3407},
  month={Apr.},
  doi={10.1109/TWC.2023.3307696}}

@ARTICLE{wu2019irs,
  author={Wu, Qingqing and Zhang, Rui},
  journal=IEEE_M_COM, 
  title={Towards Smart and Reconfigurable Environment: Intelligent Reflecting Surface Aided Wireless Network}, 
  year={2020},
  volume={58},
  number={1},
  pages={106-112},
  month={Jan.},
  doi={10.1109/MCOM.001.1900107}}

@ARTICLE{zhu2023modeling,
  author={Zhu, Lipeng and Ma, Wenyan and Zhang, Rui},
  journal=IEEE_J_WCOM, 
  title={Modeling and Performance Analysis for Movable Antenna Enabled Wireless Communications}, 
  year={2024},
  month={June},
  volume={23},
  number={6},
  pages={6234-6250},
  doi={10.1109/TWC.2023.3330887}}

@article{wong2020fluid,
  title={Fluid antenna systems},
  author={Wong, Kai-Kit and Shojaeifard, Arman and Tong, Kin-Fai and Zhang, Yangyang},
  journal=IEEE_J_WCOM,
  volume={20},
  number={3},
  pages={1950--1962},
  year={2020},
  month={Mar.},
  publisher={IEEE}
}

@article{wu2021intelligent,
  title={Intelligent reflecting surface-aided wireless communications: A tutorial},
  author={Wu, Qingqing and Zhang, Shuowen and Zheng, Beixiong and You, Changsheng and Zhang, Rui},
  journal=IEEE_J_COM,
  volume={69},
  number={5},
  pages={3313--3351},
  year={2021},
  month={May},
  publisher={IEEE}
}

@ARTICLE{ding2024pin,
  author={Ding, Zhiguo and Schober, Robert and Vincent Poor, H.},
  journal=IEEE_J_COM, 
  title={Flexible-Antenna Systems: A Pinching-Antenna Perspective}, 
  year={2025},
  volume={},
  number={},
  pages={1-1}}
\end{document}